\newif\iffullpaper   \fullpapertrue
\iffullpaper
  \documentclass[acmsmall,nonacm]{acmart}
\else
  \documentclass[acmsmall]{acmart}
\fi

\usepackage{float}
\usepackage{wrapfig}
\usepackage{needspace}
\usepackage{graphicx}
\usepackage{bm}
\usepackage{thmtools}
\usepackage{algorithm}
\usepackage{algorithmic}
\usepackage{array}
\usepackage{verbatim}
\usepackage{bm}
\usepackage{xspace}
\usepackage{comment}
\usepackage{subfig}
\usepackage{tikz}
\usetikzlibrary{positioning}
\usetikzlibrary{calc}
\usetikzlibrary{arrows.meta}
\usetikzlibrary{fit, positioning}
\usetikzlibrary{shapes.geometric}
\usetikzlibrary{decorations.pathreplacing}
\usetikzlibrary{backgrounds}

\usepackage{xcolor}

\newcommand{\apref}[2]{\iffullpaper#1\else#2\fi}

\AtBeginDocument{%
  }

\newtheorem{observation}{Observation}

\newcommand{\dl}{\textit{ :- }}
\newcommand{\norm}[2]{\left\lVert #1 \right\rVert_{#2}}

\newcommand{\trop}{\text{\sf Trop}}
\newcommand{\true}{\text{\sf true}}
\newcommand{\false}{\text{\sf false}}

\newcommand{\N}{\mathbb N} 
\newcommand{\R}{\mathbb R} 

\newcommand{\Name}{\text{\sf Datalog}^\circ}
\newcommand{\cd}{\text{ :- }}

\newcommand{\Dom}{\textsf{Dom}}

\allowdisplaybreaks
\newcommand{\multisetminus}{\mathbin{\setminus\mkern-6mu\setminus}}

\begin{document}

\title{Optimal Convergence of Iterative Methods for Datalogo}

\author{Simon Frisk}
\email{simon.frisk@wisc.edu}
\affiliation{
  \institution{University of Wisconsin-Madison}
  \city{Madison}
  \state{Wisconsin}
  \country{USA}
}

\author{Sungjin Im}
\email{sjin.im@gmail.com}
\affiliation{
  \institution{University of California, Santa Cruz}
  \city{Santa Cruz}
  \state{California}
  \country{USA}
}

\author{Paraschos Koutris}
\email{paris@cs.wisc.edu}
\affiliation{
  \institution{University of Wisconsin-Madison}
  \city{Madison}
  \state{Wisconsin}
  \country{USA}
}

\author{Benjamin Moseley}
\email{moseleyb@andrew.cmu.edu}
\affiliation{
  \institution{Carnegie Mellon University}
  \city{Pittsburgh}
  \state{Pennsylvania}
  \country{USA}
}

\author{Hung Ngo}
\email{hung.ngo@relational.ai}
\affiliation{
  \institution{RelationalAI}
  \city{Berkeley}
  \state{California}
  \country{USA}
}

\author{Kirk Pruhs}
\email{kirk@cs.pitt.edu}
\affiliation{
  \institution{University of Pittsburgh}
  \city{Pittsburgh}
  \state{Pennsylvania}
  \country{USA}
}

\author{Hangdong Zhao}
\email{hangdongzhao@microsoft.com}
\affiliation{
  \institution{Gray Systems Lab, Microsoft}
  \city{Redmond}
  \state{Washington}
  \country{USA}
}

\renewcommand{\shortauthors}{}

\begin{abstract}
  $\Name$ has been introduced as an extension to Datalog that increases expressiveness, yet retains simple least fixpoint semantics and admits optimization
  techniques such as semi-naïve evaluation and demand transformation. $\Name$ accomplishes
  this by generalizing the {\em or} and {\em and} operators of Datalog to addition and
  multiplication over a semiring.

  Finding a (minimal) fixpoint of a $\Name$ program is equivalent to finding a solution to a
  system of polynomial equations over the underlying semiring. Solving these systems of
  polynomial equations is not only a fundamental problem in the theory of $\Name$, but also has
  many applications in computer science, such as in databases, program analysis, and
  optimization.

  This paper resolves a key open problem in the theory of $\Name$: we prove a tight upper
  bound on the number of steps until convergence of iterative methods for solving these
  polynomial equation systems over a commutative $p$-stable semiring. In particular, we show that the
  number of steps until convergence is $O((p+1)n)$ where $n$ is the output size of the
  $\Name$ program.

  As the number of steps until convergence is only a proxy for the runtime of
  $\Name$ evaluation, we also consider the number of semiring operations used
  and show that, for a natural class of algorithms, $O((p+1)mn)$ is
  a tight upper bound, where $m$ is the maximum number of semiring operations per iteration.
\end{abstract}

\iffullpaper
\authorsaddresses{}
\fi

\maketitle

\section{Introduction}
\label{sec:intro}

\subsection{$\Name$ and systems of polynomial equations over a semiring}

Computations in machine learning and graph analytics often require expressing iterative or
recursive computations involving aggregates such as $\min, \max, +$, etc. For example,
let $T(X,Y)$ denote the shortest path length between nodes $X$ and
$Y$ in a graph, and let $E(X,Y)$ denote the edge weights in the
graph. Then the all-pairs shortest paths (APSP) problem can be expressed naturally as:
\begin{align}
  T(X,Y) &= \min \{ E(X,Y), \min_Z \ (T(X,Z) + E(Z,Y)) \}.
    \label{eqn:APSP}
\end{align}
Writing these types of problems in traditional Datalog, even with aggregates, is clunky and faces semantic hurdles: a minimal fixpoint from a set
containment perspective is not necessarily the (shortest-path) fixpoint we are looking for.
There have been multiple attempts to extend Datalog to
express these types of problems
(see~\cite{DBLP:journals/jacm/KaminskiKGMH22,DBLP:journals/tplp/GrecoZ01} and references
therein).

Motivated by these problems, and the desire to retain the elegant least fixpoint
semantics of Datalog, Abo-Khamis et al.~\cite{Khamis0PSW22} introduced $\Name$, a language
that extends Datalog with aggregates and recursion over an arbitrary semiring. The
formulation of a $\Name$ program is similar to a Datalog program, with two
modifications: first, we replace $(\vee, \wedge)$ with $(\oplus, \otimes)$, where $\oplus$
and $\otimes$ are addition and multiplication operators in a
semiring\footnote{To not impede the flow of the introduction, we postpone defining
background concepts, like semirings, to Section~\ref{sec:prelim}.} $\bm S = (S, \oplus,
\otimes, 0, 1)$. Second, relations are replaced by functions whose range is the domain $S$.
These functions are also called {\em $\bm S$-relations}.\footnote{They are also called
$K$-relations in the seminal work of Green et al.~\cite{DBLP:conf/pods/GreenKT07}.}
Continuing Example~\eqref{eqn:APSP}, $E$ and $T$ are functions $\Dom \times \Dom \to
S$, where $\Dom$ is the domain of $X, Y$ and $Z$. In $\Name$, the rule becomes:
\begin{align}
  T(X,Y) & \cd E(X,Y) \oplus \bigoplus_Z T(X,Z) \otimes E(Z,Y),
  \label{eqn:datalogo:apsp}
\end{align}
By changing the semiring $\bm S$ and the underlying data ($E$), the $\Name$
rule~\eqref{eqn:datalogo:apsp} captures a wide range of problems. For example, when $\bm S$
is the Boolean semiring $(\{\false,\true\}, \vee, \wedge, \false, \true)$, the
rule~\eqref{eqn:datalogo:apsp} computes the transitive closure of $E$:
\begin{align}
  T(X,Y) & \cd E(X,Y) \vee \exists_Z \ (T(X,Z) \wedge E(Z,Y))
  \label{eqn:transitive_closure}
\end{align}
When $\bm S$ is the tropical semiring $\trop_+ = (\R_+ \cup \{\infty\}, \min, +, \infty,
0)$, and $E$ is the edge-weight function of a graph, the rule becomes~\eqref{eqn:APSP},
which solves the APSP problem.

The generalization from Datalog to $\Name$ has an intermediate step that was widely
recognized since the 1960s, the {\em algebraic path
problem}~\cite{DBLP:journals/cacm/Floyd62a, DBLP:journals/jacm/Warshall62, MR1059930,
DBLP:journals/jacm/Tarjan81a, DBLP:journals/tcs/Lehmann77, DBLP:journals/computing/Rote85}.
Researchers have long recognized that the transitive closure problem and the all-pairs
shortest paths problem can be expressed similarly by abstracting out the underlying
semiring. This creates the algebraic path problem which admits
common algorithms, such as the Floyd-Warshall
algorithm~\cite{DBLP:journals/cacm/Floyd62a,DBLP:journals/jacm/Warshall62}. These algorithms are described in standard algorithms textbooks~\cite{CLRS}.

$\Name$ generalizes the algebraic path problem by having EDBs and IDBs as $\bm S$-relations. $\Name$ can be evaluated via a natural
iterative algorithm: start by assigning all IDB-function values the additive
identity $\bm 0$, and then repeatedly apply the rules until a fixed point is
reached. This is sometimes called the ``na\"ive evaluation'' algorithm, in other contexts known as
Jacobi iteration, Gauss-Seidel iteration, and Kleene
iteration~\cite{semiring_book,DBLP:journals/computing/Rote85}. The output is the set of all
non-$\bm 0$ valued IDB tuples with associated semiring values. Throughout the paper, we denote by $n$ the number of IDB tuples in the output.

Another way to view the iterative evaluation of $\Name$ programs is as a method for
solving a system of polynomial equations over the semiring $\bm S$. The equation system
comes from a ``grounding'' of the $\Name$ program. As a brief example, consider the $\Name$
program~\eqref{eqn:datalogo:apsp}.
Suppose the domain of $X, Y, Z$ is $\Dom = [4]$, and the EDB mapping $E$ is defined as:
\begin{align*}
  E(1,2) &= e_{12}, \quad E(2,3) = e_{23}, \quad E(3,4) = e_{34}, \quad
  E(i,j)  = \bm 0 \text{ for all other } (i,j) \in \Dom \times \Dom.
\end{align*}
Here, $e_{12}, e_{23}, e_{34}$ are elements of the semiring $\bm S$.
Then, thinking of $T(i,j)$ as an unknown variable $X_{ij}$, the grounding of the $\Name$ program
results in a system of $4\times 4 = 16$ polynomial equations:
\begin{align}
    X_{ij} & = e_{ij} \oplus \bigoplus_{k=1}^4 X_{ik} \otimes e_{kj},
\end{align}
where $e_{ij}=\bm 0$ other than for $(i,j) \in \{(1,2), (2,3), (3,4)\}$.
Most of these variables $X_{ij}$ will be $\bm 0$ because the EDB $E$ is sparse.
The effective system of equations is:
\begin{align}
    X_{12}
    & = e_{12},
    && X_{13} = X_{12} \otimes e_{23},
    & X_{14} = X_{13} \otimes e_{34}, \label{eqn:system:of:equations} \\
    X_{23}
    &= e_{23},
    && X_{24} = X_{23} \otimes e_{34},
    & X_{34} = e_{34}. \nonumber
\end{align}
More generally, the system of equations is of the form $\bm X = \bm f(\bm X)$,
where $\bm X$ is a vector of $n$ variables over $S$, and $\bm f : S^n \to S^n$ is a vector
of polynomial functions over the semiring $\bm S$. The function $\bm f$ is {\em linear} if
there is no monomial of degree greater than $1$ in any of the equations.

Solving  systems of polynomial equations is a fundamental task
and has applications in automata theory, databases, program analysis, optimization, and many
other areas~\cite{MR1470001,DBLP:conf/popl/CousotC77,MR1728440,
MR556411,MR584516,MR526496,MR1059930, DBLP:conf/lics/HopkinsK99,
DBLP:journals/tcs/Lehmann77, semiring_book,MR609751}. Under  assumptions about the
semiring $\bm S$ and/or the function $\bm f$, there are non-iterative algorithms such as the
Newton method~\cite{DBLP:journals/jacm/EsparzaKL10,DBLP:conf/popl/RepsTP16} or generalized
Gaussian elimination~\cite{DBLP:journals/tcs/Lehmann77,DBLP:journals/computing/Rote85}. In
general, however, the iterative method is the most {\em robust}, as it is applicable with the least assumptions about the input, and it does not
require materializing a quadratic-sized matrix. In relational database systems,
iterating to fixpoint is also the main method for evaluating $\Name$.
(Semi-naïve evaluation only reduces the complexity of each step; it does not
change the number of iterations.)

This motivates the two problems we address in this paper.
First, bounding the number of iterations to convergence, $\min_k \bm f^{(k)}(\bm 0)
= \bm f^{(k+1)}(\bm 0)$, is fundamental both to theory and practice of $\Name$, 
and related problems in other areas. In theory, it is a fundamental result that (positive)
Datalog programs can be evaluated in PTIME~\cite{DBLP:books/aw/AbiteboulHV95}: this follows because each iteration derives at least one new fact, bounding the iterations by the output size $n$. 
One may hope for a smooth generalization of this result to $\Name$ programs over a
$p$-stable semiring, since the Boolean semiring from Datalog is $0$-stable. 
In practice, these bounds are useful to detect run-away computations.

Second, the number of iterations is only a proxy for the runtime of algorithms.
For example, the binary-recursive formulation of transitive closure
\begin{align}
  T(X,Y) & \cd E(X,Y) \vee \exists_Z \ (T(X,Z) \wedge T(Z,Y))
  \label{eqn:binary:transitive_closure}
\end{align}
converges in exponentially fewer iterations than the linear formulation~\eqref{eqn:transitive_closure},
yet it is not faster.
It is therefore important to study the actual runtime of the fixpoint search algorithm.

\subsection{Previous Results}

A semiring $\bm S = (S, \oplus, \otimes, 0, 1)$ is $p$-stable if for all elements $u \in S$
we have $u^{(p)}=u^{(p+1)}$, where $u^{(p)} := 1 \oplus u \oplus u^2 \oplus \cdots \oplus
u^{p}$. A semiring  $\bm S$ is stable if there exists a constant $p$ such that $\bm S$ is
$p$-stable. The justification for restricting semirings to being stable arises from the fact
that $\Name$ programs over non-stable semirings may not always converge to a
fixpoint~\cite{Khamis0PSW22}.

A series of previous results have established bounds on the iteration count for
iterative evaluation of $\Name$ programs. In~\cite{Khamis0PSW22}, the authors showed that
the number of iterations is $O((p+2)^n)$ in general, and
$O((p+1)^n)$ for linear $\Name$. If $p=0$ they proved a bound of $O(n)$ steps. For the
special cases of the $\trop_p$ semiring and linear $\Name$ programs, the number of iterations
to convergence is bounded by $(p+1)n-1$, which is tight~\cite{Khamis0PSW22}.

The exponential gap between $p=0$ and $p>0$ was improved in~\cite{DatalogoLinear,
DatalogoPolytime}. In \cite{DatalogoLinear} it was shown that if the semiring is
$p$-stable with $L$ elements in the semiring domain, then the iterative evaluation of all
linear $\Name$ programs converges after $O(\min((p+1) n^3, (p+1) n \lg L))$ steps.
In~\cite{DatalogoPolytime}, the non-linear case was considered, and the bound was shown
to be $O((p+1)kn^4\log(kd))$ steps, where
$k$ is the input size and $d$ is the degree of the polynomial $\bm f$.

\subsection{Our Contributions}

Our research is motivated by the gaps in bounds established by previous work, which are not tight, and have some surprising properties: there is a gap between linear and non-linear $\Name$ programs, and the bounds for non-linear $\Name$ programs depended on the input size.
Moreover, we observed a bug in the $\Omega((p+1)n^3)$ lower bound from~\cite{DatalogoLinear}.%
\footnote{The authors of~\cite{DatalogoLinear} concur with this assessment.} 

Our {\bf first contribution} is to show that the iterative
evaluation of linear $\Name$ programs actually converges in a linear number of iterations.

\begin{theorem}
\label{thm:linearmain}
The iterative evaluation of every linear $\Name$ program with $n$ IDB tuples over a $p$-stable commutative
semiring converges in at most  $(2p+1)n$ iterations.
\end{theorem}

Our {\bf second contribution} is to show a similar upper bound for general $\Name$.

\begin{theorem}
    \label{thm:nonlinearmain}
    The iterative evaluation of every $\Name$ program with $n$ IDB tuples over a $p$-stable commutative semiring
    converges in at most  $(2p+2)n$ iterations.
\end{theorem}

Together with the lower bounds of
$(p+1)n-1$ shown in \cite{Khamis0PSW22}, this establishes a tight bound within a factor of
$2+\epsilon$. Furthermore, this shows that the iteration count for $\Name$ and Datalog is
within a factor of $(2p+2)$. Since $p$ is often a small
constant, this means convergence rates are asymptotically the same.
This also shows that the convergence rates of nonlinear and linear $\Name$ programs are asymptotically the same.

Our next research question is to consider if other algorithms than iteration can more efficiently evaluate $\Name$ programs. Here, iteration count is not a good complexity measure, as time per
iteration is important, and the concept of an ``iteration'' may not be well
defined. A natural candidate complexity measure is the number of semiring
operations used, which we call the {\em semiring complexity}. If $m$
is the maximum number of semiring operations per iteration\footnote{Here $m$ depends on the grounding of the program, i.e. how the program is transformed into a system of polynomial equations. Different groundings can give different $m$.}, then our upper bounds show that the worst-case semiring complexity of iterative evaluation is $O((p+1)nm)$.

Many standard algorithms for evaluating linear Datalog programs fall within a class
of algorithms we call {\em relaxation algorithms} -- this terminology is also used in the context of computing shortest paths~\cite{CLRS}. At a high level, a
relaxation algorithm consists of a sequence of relaxations, where each relaxation has roughly the form $x = x
\oplus a \otimes y$, where $x$ and $y$ are variables, and $a$ is a semiring value specified
in the $\Name$ program. Examples of relaxation algorithms include na\"{\i}ve and
semi-na\"{\i}ve evaluation~\cite{DBLP:books/aw/AbiteboulHV95}, as well as Dijkstra's
algorithm for shortest paths (other examples are found in \cite{GroundingBasedApproach}).

Our {\bf third contribution} is a lower bound on the semiring complexity of linear $\Name$
programs, which shows that our upper bound $O((p+1)mn)$ is optimal among relaxation
algorithms.

\subsection{Paper Organization}

Section~\ref{sec:prelim} gives necessary background context. Section~\ref{sec:absorption} explains our upper bound analysis technique, which we call the \textit{absorption
replacement process}. The upper bound
for linear $\Name$ is given in Section~\ref{sec:linearupperbound}, while the upper bound for general/nonlinear $\Name$ is given in Section~\ref{sec:generalupperbound}. Our lower bounds are given in Section~\ref{sec:lowerbounds}.
Due to space limitations, some details, particularly proofs,
have been moved to the appendices.

\section{Preliminaries}
\label{sec:prelim}

\subsection{Multisets}

A multiset $M$ is a collection of objects that can contain multiple copies of the same
object $x$. We assume that the elements $x$ are distinguishable so that we can
distinguish the first copy of $x$, the second copy of $x$, etc. By $x \in M$, we mean that
$M$ contains at least one of $x$. We use $M \uplus N$ to denote the multiset union of
multisets $M$ and $N$, where the number of copies of an element $x$ in $M \uplus N$ is the
maximum of the number of copies of $x$ in $M$ and the number of copies of $x$ in $N$.
The sum $M+N$ is the multiset where the number of copies of an element is the sum of the
number of copies in $M$ and $N$.
We use
$M \multisetminus N$ to denote the multiset difference of multisets $M$ and $N$, where the
number of copies of an element $x$ in $M \multisetminus N$ is the maximum of zero, and the
number of copies of $x$ in $M$ minus the number of copies of $x$ in $N$. In most applications of multiset difference, we will know a priori that the number of copies of
$x$ in $M$ is at least the number of copies of $x$ in $N$.

\subsection{Semirings}
A {\em semiring} is a tuple $\bm S = (S, \oplus, \otimes, 0, 1)$ where $S$ is a set called the \textit{domain} of the semiring, $\oplus$ and
$\otimes$ are binary operators on $S$, $(S, \oplus, 0)$ is a commutative monoid (meaning
$\oplus$ is commutative and associative, and $0$ is the identity for $\oplus$), $(S,
\otimes, 1)$ is a monoid (meaning $\otimes$ is associative, and $1$ is the identity for
$\otimes$), $a\otimes 0 = 0 \otimes a = 0$ for every $a \in S$, and $\otimes$ distributes
over $\oplus$. $\bm S$ is said to be {\em commutative} if $\otimes$ is commutative. A
standard reference on semirings is~\cite{semiring_book}.

We give some more definitions.
Given an element $u$ of the semiring, and $p$ a nonnegative integer, 
we define $u^{(p)} := 1 \oplus u \oplus u^2 \oplus \cdots \oplus u^{p},$ where $u^{i}
:= u \otimes u \otimes \cdots \otimes u$ ($i$ times). 
We define $pu$ as $\oplus_{i=1}^p u$, where the empty sum is by convention the
multiplicative unity 1. An element $u\in S$ is {\em
$p$-stable} if $u^{(p)}=u^{(p+1)}$, and a semiring $\bm S$ is {\em $p$-stable} if every
element $u \in S$ is $p$-stable. We say that $y \in S$ {\em absorbs} $x \in S$ if $x \oplus y =
y$. Let $M$ be a multiset of elements from $S$. Then we define $\oplus(M)$ to be the sum over all elements of $M$. 
We say that $M$ {\em absorbs} $x \in S$ if
$\oplus(M)$ absorbs $x$. The following fact, which follows from $p$-stability of $1$, will be
useful:
\begin{observation}
\label{obs:coef-limit}
Given a $p$-stable commutative semiring $\bm S = (S, \oplus, \otimes, 0, 1)$, for any $u \in
S$, we have $(p+1) u = (p+2) u$.
\end{observation}

\subsection{$\Name$ and Polynomial Systems}

As we briefly explained in Section~\ref{sec:intro}, a $\Name$ program can be viewed as a
polynomial map $\bm f : S^n \to S^n$, with component functions $f_i : S^n \to S$, $i \in
[n]$, where $S$ is the domain of the semiring, and $n$ is the number of grounded IDBs (the output size of the program). This polynomial system is sometimes called a \textit{grounding}.
An example of a grounding is shown in the introduction.
Roughly speaking, the grounding can be obtained by \textit{applying} the $\Name$ rules to the grounded IDBs and EDBs, which are the possible output and input tuples, respectively. This means that instead of having rules that state a general derivation, we have instantiated derivations for possible combinations of input and output.
For example, to ground a reachability program, for each vertex $v$ we have a variable $x_v$ representing its reachability, which in this case is boolean. We write the polynomial system by, for each vertex $v$, writing an equation for its reachability, which is the sum: $x_v=x_{v_1} \oplus x_{v_2} \oplus \cdots \oplus x_{v_k}$, where $v_1, v_2, \ldots, v_k$ are the vertices that can reach $v$. If the reachability program has a ``starting node'', we can add a term in the sum for that node, which is a constant EDB variable with value \textit{true}.
Grounding a general $\Name$ program is described in detail in previous work \cite{Khamis0PSW22,GroundingBasedApproach}. We remark that there can be multiple ways to ground a $\Name$ program, and this choice can affect the complexity of iterative evaluation, as it impacts the cost of each iteration.

In the example shown
on line~\eqref{eqn:system:of:equations}, we can assume $n=6$ and $\bm f : S^6 \to S^6$, as there
are at most $6$ grounded IDBs that could be non-zero. 
After $q$ iterations of iterative evaluation of a $\Name$ program,
the semiring values
of the tuples in the output are given by 
$\bm f^{(q)}(\bm 0) $, where $\bm 0 = (0,\dots, 0)$ is the all-zero vector,
and $\bm f^{(q)}$ is the composition of $\bm f$
with itself $q$ times.
Thus the number of iterations until convergence is
the minimum  $q$ such that
$\bm f^{(q)}(\bm 0) = \bm f^{(q+1)}(\bm 0)$.

\subsection{$\Name$ and context-free languages}

Given this polynomial function $f$, one can then define, for each component function $f_i$, a context-free grammar (CFG). 
We will explain how this grammar arises with an example: consider the  polynomial function $\bm f: S^2 \to S^2$,
where $f_1(X,Y) = (a\otimes X \otimes Y) \oplus (b\otimes Y) \oplus c$,
$f_2(X,Y) = c \otimes X \oplus d$,
$a, b, c, d$ are elements of the semiring, and $X$ and $Y$ are IDBs. 
This $f$ can be naturally
modeled by a CFG over non-terminals $\{X, Y\}$ and terminals $\{a,b,c,d\}$ with
the following derivation rules:
\begin{align}
    X &\rightarrow aXY \mid bY \mid c &&
    & Y &\rightarrow cX \mid d
    \label{eqn:example:cfg}
\end{align}

A concept from context-free grammars that we will use heavily is the notion of a \textit{derivation tree}, sometimes also called a \textit{parse tree}. A derivation tree is a tree where the internal nodes are non-terminal symbols and the leaves are terminal symbols. The children of a non-terminal are described by the production rules of the grammar. Derivation trees have been heavily used in previous research when analyzing different properties of Datalog programs \cite{belowabovewhyprovenancedatalog}. 

Composing $\bm f$ with itself can now be seen as looking at all derivation trees 
of this CFG that have depth at most 2.
The value of $f_1^{(2)}(\bm 0)$ is $c \oplus ( b \otimes d) \oplus (a\otimes c \otimes d )$, which corresponds to the sum of the product of the leaves across all derivation trees with root $X$ and depth at most 2 (see Figure~\ref{fig:example:cfg}).
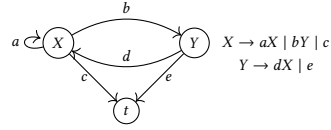
\begin{figure}[th]
\begin{minipage}[t]{0.44\textwidth}
  \caption{Derivation trees of depth $\leq 2$ for CFG~\eqref{eqn:example:cfg}}
  \label{fig:example:cfg}
  \centering
  \begin{tikzpicture}[scale = 0.6, every node/.style={transform shape}]
    \node [] at (0,1) (x1) {$X$};
    \node [] at (0,0) (c) {$c$};
    \draw [] (x1) to (c);
    \node [] at (1.5,1) (x2) {$X$};
    \node [] at (1,0) (b) {$b$};
    \node [] at (2,0) (y) {$Y$};
    \node [] at (2,-1) (d) {$d$};
    \draw [] (x2) to (b);
    \draw [] (x2) to (y);
    \draw [] (y) to (d);
    \node [] at (3.5,1) (x3) {$X$};
    \node [] at (3,0) (a) {$a$};
    \node [] at (4,0) (y2) {$Y$};
    \node [] at (4,-1) (d2) {$d$};
    \node [] at (3.5,0) (x4) {$X$};
    \node [] at (3.5,-1) (c2) {$c$};
    \draw [] (x3) to (a);
    \draw [] (x3) to (y2);
    \draw [] (y2) to (d2);
    \draw [] (x4) to (c2);
    \draw [] (x3) to (x4);
  \end{tikzpicture}
\end{minipage}%
\hfill
\begin{minipage}[t]{0.55\textwidth}
  \caption{Example of a transition graph}
  \label{fig:example:linear:cfg}
  \centering
  \begin{tikzpicture}[scale = 0.6, every node/.style={transform shape}]
    \node [circle, draw] at (0,1) (x) {$X$};
    \node [circle, draw] at (3,1) (y) {$Y$};
    \node [circle, draw] at (1.5,-0.5) (t) {$t$};
    \draw [->] (x) to [bend left] node[above] {$b$} (y);
    \draw [->] (y) to [bend left] node[above] {$d$} (x);
    \draw [->] (x) edge[loop left] node {$a$} ();
    \draw [->] (x) to node[left] {$c$} (t);
    \draw [->] (y) to node[right] {$e$} (t);
     \node at (4.8,1) (x) {$X \rightarrow aX \mid bY \mid c$};
     \node at (4.8,0.5) (x) {$Y \rightarrow dX \mid e$};
  \end{tikzpicture}
\end{minipage}
\end{figure}

This construction dates back to Pilling~\cite{MR339568} and Hopkins and
Kozen~\cite{DBLP:conf/lics/HopkinsK99}. In general, the context-free grammar for the $f_s$ component
function of $\bm f$ is defined by the 4-tuple $G_s
= (V, S, R, X_s)$ where $V$ is the set of variables (one for each output grounded IDB), $S$
is the set of terminal symbols (one for each element of the domain of the semiring that appears in the program), $R \subset V \times (V \cup S)^*$ is a finite set of
derivation rules, and $X_s$ is the start symbol that is the root of all derivation trees (which is
the $s$th output grounded IDB).
The  {\em yield} of a derivation tree $T$ 
of such a grammar $G_s$ is the product of its leaves,
denoted by $Z(T)$.
For any non-negative integer $q$, and any $s \in [n]$, let $\mathcal T_s^q$
denote the set of all derivation trees for grammar $G_s$
that have depth  $ q$. It is known~\cite{MR339568,DBLP:conf/lics/HopkinsK99} that evaluating the $q$th composition 
of $f_s$ with itself, evaluated at zero, is the sum of the yields of
all derivation trees of depth at most $q$, that is:
\begin{align}
    \label{eqn:general-f-trees}
f_s^{(q)}(\bm 0) = \bigoplus_{i=0}^q \bigoplus_{T \in \mathcal T_s^i} Z(T)
\end{align}
A $\Name$ program converges in $q$ steps if $q$ is the
minimum integer such that $f_s^{(q)}(\bm 0)=f_s^{(q+1)}(\bm 0)$.

\subsection{Linear $\Name$}

A $\Name$ program is linear if each of its component functions $f_s$ is linear, meaning that
in each monomial of $f_s$ at most one variable appears. To
restate this in terms of the CFG $G_s$, 
linearity means that all derivation rules
are of the form $X_i \to cX_j$ or $X_i \to c$, where $c$ is an element of the ground set of
the semiring $\bm S$ (the value $c$ is the semiring product of the semiring values of the EDB atoms). For example, the grammar in Figure~\ref{fig:example:linear:cfg} is linear.

For linear $\Name$ programs there is an equivalent graph-theoretic way to represent the derivation trees that is more
convenient to analyze~\cite{Khamis0PSW22,DatalogoLinear}. This representation is a  ``transition'' graph $G = (V \cup \{t\},E)$ where the vertex set $V$ is the set of
variables from the $G_s$ grammars, and $t$ is a special ``sink'' vertex. For every derivation rule $X_i
\to cX_j$, there is a directed edge $(X_i, X_j) \in E$ that
is labeled with $c$. For every derivation
rule $X_i \to c$, there is a directed edge $(X_i, t) \in E$ labeled with $c$. We use $x_e$ to
denote the label (semiring value) that is assigned to the edge $e$.
An example of a transition graph is shown in Figure~\ref{fig:example:linear:cfg}.
Let 
$\mathcal{W}_{s,t}^k$  denote the set of all walks from $s$ to $t$ of length $k$ in $G$.
There is then a one-to-one correspondence between the derivation trees 
of $G_s$ with depth $k$  and elements of $\mathcal{W}_{s,t}^k$. 
One can then express $f_s^{(q)}(\bm 0)$ as
\begin{align}
    f_s^{(q)}(\bm 0) =  \bigoplus_{k=1}^q \bigoplus_{W \in {\mathcal{W}^k_{s,t}}}  \bigotimes_{e \in W} x_e,
\label{prop:linear:convergence:criterion}
\end{align}
In other words, $f_s^{(q)}(\bm 0)$ is the sum over all walks $W$ 
from $s$ to $t$ in $G$, with  at most $q$ edges, of the
product of the labels on the edges of $W$.
We use extensively the following characterization of edges that form a walk~\cite{Corberan}. 

\begin{observation}
\label{obs:walkcharactorization}
    A collection $C$ of edges can be ordered to form a walk $W$ from $s$ to $t$ if and only
    if the edges in $C$ are connected (when viewed as undirected edges), the out-degree
    of $s$ is one greater than the in-degree of $s$, the in-degree of $t$ is one greater than
    the out-degree of $t$, and for all other vertices $v$ the in-degree of $v$ is equal to the
    out-degree of $v$.
\end{observation}

The `only if' direction follows from the fact that any $s$-$t$ walk traversing each edge once must enter and leave intermediate vertices an equal number of times. For the `if' direction, the degree constraints imply that $C \cup \{(t, s)\}$ satisfies the conditions for an Eulerian circuit; the connectivity of $C$ ensures this circuit covers all edges, which, upon removing $(t, s)$, yields the desired $s$-$t$ walk.

We generally adopt standard graph-theoretic terminology. 
For brevity, we will use {\em loop}  to mean a closed walk that does not
traverse any vertex more than once (instead of the more
common term ``simple cycle''). We also employ the
following path-loop decomposition of a walk as used in~\cite{DatalogoLinear}.

\begin{definition}
    Let $W$ be a walk from $s$ to $t$ in a directed graph. A {\em path-loop decomposition}
    of $W$ is a pair $(P, \mathcal L)$ where $P$ is a simple path from $s$ to $t$ and
    $\mathcal L$ is a multiset of loops. Here, the multiset of edges in $W$ is equal to
    the multiset of edges obtained by taking the (multiset) sum  of $P$ and the edges in $\mathcal L$.
\end{definition}

We sometimes use the term \textit{loop decomposition} referring to only $\mathcal L$.

\subsection{Nonlinear $\Name$}
\begin{figure}[th]
\begin{minipage}{0.45\textwidth}
  \caption{An example derivation tree \label{fig:transformbefore}}
  \centering
\begin{tikzpicture}[scale=0.6, every node/.style={transform shape}]
  \node [] at (0,1)   (Xroot) {$X$};
  \node [] at (-0.8,0) (aL)  {$a$};
  \node [] at (0.0,0)  (X2)  {$X$};
  \node [] at (1.2,0)  (Y)   {$Y$};
  \node [] at (0.0,-1) (a2)  {$a$};
  \node [] at (1.2,-1) (d)   {$d$};

  \draw [] (Xroot) -- (aL);
  \draw [] (Xroot) -- (X2);
  \draw [] (Xroot) -- (Y);
  \draw [] (X2) -- (a2);
  \draw [] (Y) -- (d);
\end{tikzpicture}

\end{minipage}%
\hfill
\begin{minipage}{0.45\textwidth}
  \caption{The corresponding compatible tree \label{fig:transformafter}}
  \centering
\begin{tikzpicture}[scale=0.7, every node/.style={transform shape}]
  \node [] at (0,1)   (Xroot) {$X$};
  \node [] at (-0.8,0) (aL)  {$t$};
  \node [] at (0.0,0)  (X2)  {$X$};
  \node [] at (1.2,0)  (Y)   {$Y$};
  \node [] at (0.0,-1) (a2)  {$t$};
  \node [] at (1.2,-1) (d)   {$t$};

  \draw [] (Xroot) -- node[midway, above] {$a$} (aL);
  \draw [] (Xroot) -- node[midway, above] {} (X2);
  \draw [] (Xroot) -- node[midway, above] {} (Y);
  \draw [] (X2) -- node[midway, right] {$a$} (a2);
  \draw [] (Y)  -- node[midway, right] {$d$} (d);
\end{tikzpicture}
\end{minipage}
\end{figure}

For our analysis of nonlinear $\Name$ it will be
convenient to modify the derivation trees,
in a manner similar to what we did for linear $\Name$ programs.
In particular we transform an arbitrary derivation tree $T$ into a tree $T'$, which we call a \textit{compatible} tree, such that
\begin{itemize}
\item The trees $T$ and $T'$ are isomorphic in the graph theoretic sense.
    \item The nonleaf nodes of $T'$ are identical to the nonleaf
    nodes of $T$. 
    \item Each edge of $T'$ has a label that is an element of the domain of the semiring.
    \item Each leaf of $T'$ is a new special sink symbol $t$.
\end{itemize}
 If the derivation tree $T$ has an edge $e=(X_i, X_j)$ in which 
 nonterminal $X_i$ is the parent of nonterminal $X_j$, then we simply label the edge $e \in T'$  with a multiplicative identity  of the semiring.  If the derivation tree $T$ has an edge $e=(X_i, c)$ in which 
 nonterminal $X_i$ is the parent of terminal/leaf $c$, we replace $e \in T'$ with the edge $(X_i, t)$ labeled with the semiring value $c$. Then for such a tree $T'$, we use $Z(T')$ to denote the product of the
 edge labels on $T'$. Note that  $Z(T)=Z(T')$. 
See Figures~\ref{fig:transformbefore} and \ref{fig:transformafter} for illustration of this transformation. 

\subsection{Proving Convergence with Absorption}
\label{subsec:absorptionreplacement}

We now explain how one can use the absorption properties
of stable semirings to prove convergence.
Recall the number of steps $q$ for $\Name$ to converge is the minimum $q$ such that $f_s^{(q)}(\bm 0)=f_s^{(q+1)}(\bm 0)$. 
Further recall that 
$f_s^{(q)}(\bm 0) = \bigoplus_{i=0}^q \bigoplus_{T \in \mathcal T_s^i} Z(T)$,
where $\mathcal T_s^i$ is the collection of derivation trees rooted
at $s$ with depth $i$. Thus to prove convergence by $q$ steps
it is sufficient to show that $\bigoplus_{T \in \mathcal T_s^{q+1}} Z(T)$ is absorbed by
$\bigoplus_{i=0}^q \bigoplus_{T \in \mathcal T_s^i} Z(T)$.
That is 
\begin{align}
    \bigoplus_{T \in \mathcal T_s^{q+1}} Z(T) + \bigoplus_{i=0}^q \bigoplus_{T \in \mathcal T_s^i} Z(T)=\bigoplus_{i=0}^q \bigoplus_{T \in \mathcal T_s^i} Z(T)
\end{align}
In turn, 
it is sufficient to show that each
$T \in \mathcal T_s^{q+1}$ is absorbed by 
$\bigoplus_{i=0}^q \bigoplus_{T' \in \mathcal T_s^i} Z(T')$, that is:
\begin{align}
Z(T) + \bigoplus_{i=0}^q \bigoplus_{T' \in \mathcal T_s^i} Z(T')=\bigoplus_{i=0}^q \bigoplus_{T' \in \mathcal T_s^i} Z(T')
\end{align}
We note a slight abuse of notation: we state that $T$ is absorbed rather than specifying the exact relation. Similarly, we may say that a collection of trees absorbs a tree. In turn it is sufficient to show that there exists
a subset $A$ of $\bigcup_{i=1}^q \mathcal T_s^i$ that
absorbs $T$, that is:
\begin{align}
Z(T) + \bigoplus_{T' \in A} Z(T')= \bigoplus_{T' \in A} Z(T')
\end{align}

The prior work  \cite{DatalogoLinear,DatalogoPolytime} 
identified the set of trees $A$ in the following manner.
First they showed that if $q$ is sufficiently big,
then the derivation tree $T$ must contain 
$p+1$ disjoint occurrences $S_1, \ldots, S_{p+1}$ 
of some derivation tree $S$, where  
removing any arbitrary collection of these copies
of $S$ in $T$ will result in a smaller (valid) derivation tree. 
Then the set $A$ consists of $p+1$ derivation trees,
where tree $i$ is $T$ with subtree occurrences
$S_1, S_2, \ldots, S_i$ deleted. 
The fact that $T$ is absorbed by the set $A$ then
follows from the stability of the semiring element
$Z(S)$.

\section{Absorption Replacement Process}
\label{sec:absorption}

In this section we explain and show correctness of the technique we use to prove convergence,
which we call the absorption replacement process.
This technique is substantially different from the process used in the prior work~ \cite{DatalogoLinear,DatalogoPolytime},
in which case the cardinality of the set $A$ which absorbs a tree can be exponentially large in $n$. 
First, we will establish some properties of stable commutative semirings which 
will be useful in our analysis. Some proofs have been moved to Appendix \ref{appendix:replacementprocess}. 
We start with a key lemma we call the \textit{Squeeze Lemma}.

\begin{restatable}[Squeeze Lemma]{lemma}{squeezelemma}
\label{lem:squeeze}
Consider a $p$-stable commutative semiring $\mathbf{S}$. Let $x$, $y$ and $z$ be elements
from the domain of $\mathbf S$ such that $x = x \oplus y \oplus z$. Then $x = x  \oplus
y$.
\end{restatable}

The Squeeze Lemma can be seen as a corollary of the fact that stable semirings are
naturally ordered (see  Theorem 5.2 of~\cite{Khamis0PSW22}). The following corollary of
the Squeeze Lemma states that the semiring value associated with an output tuple can never
transition from being nonzero to zero.

\begin{restatable}{corollary}{nonzerocorollary}
\label{cor:nonzero}
Consider a $p$-stable commutative semiring $\mathbf{S}$. If $y$ and $z$ are elements from the domain of $\bm S$ that are each nonzero, then $z \oplus y \ne 0$.
\end{restatable}

The next lemma, which is central to our analysis, states that a monomial of degree at least $p+1$  is absorbed
by a sum of all possible products of up to $p$ multiplicands from the monomial. Let $\bm S$ be a
semiring and consider a set $\{u_1, \dots, u_r\}$ of $r$ elements of $\bm S$. For any
exponent vector $\bm m \in \N^r$, define $\bm u^{\bm m} := \bigotimes_{i=1}^r u_i^{m_i}$. We denote by $\binom{m}{m_1,\dots,m_r}$ the multinomial coefficient, where $m=\sum_{i=1}^rm_i$.

\begin{restatable}[Parallel Absorption]{lemma}{parallelabsorptionlemma}
\label{lemma:Semiring-AbsorptionLemma}
Consider a $p$-stable commutative semiring $\mathbf{S}$. Consider the product  $x = \bm u^{\bm m}$
where $\norm{\bm m}{1} > p$. Then $x$ is absorbed by
$\bigoplus_{i=0}^p \bigoplus_{\bm k : \norm{\bm k}{1}=i} \binom{i}{k_1,\dots,k_r} \bm u^{\bm k}$.
\end{restatable}

We next show that absorption has a type of transitivity property.

\begin{restatable}[Absorption Transitivity]{lemma}{absorptiontransitivity}
\label{lem:absorbfoundation}
Consider a $p$-stable commutative semiring $\mathbf{S}$. Let $s$ be an element from the domain of $\,\mathbf S$, let $A$ and $B$ be two multisets of elements from the domain of 
$\,\mathbf S$, let $a$ be an element of $A$ that is not an element of $B$, and let $C = ( A
\multisetminus \{a\} ) \uplus B$. Then if $s$ is absorbed by $\oplus(A)$ and $ a$ is
absorbed by $\oplus(B) $ then $s$ is absorbed by $\oplus(C) $.
\end{restatable}

We are now ready to describe the \emph{absorption replacement process} that we will use to
show that a deep derivation tree $T$ is absorbed by shallower trees. We maintain a set $A$
of derivation trees that is initialized to contain only the tree $T$. Then, in each
iteration of this process, a deep tree $T'$ is removed from $A$, and replaced by a collection
$B$ of derivation trees that collectively absorb $T'$. 
To describe the process, we will describe how to produce $B$ from $T'$. 
In the analysis of this process, we need to argue that
the process terminates and that all resulting trees in $A$ are shallow.
Key building blocks here will be the Absorption Transitivity Lemma and the Parallel Absorption Lemma.
When we consider linear $\Name$ programs, we can simplify and consider walks instead of general
derivation trees.

\section{Upper Bounds for Linear Datalogo Programs}
\label{sec:linearupperbound}

The goal in this section is to establish Theorem \ref{thm:linearmain},
that linear $\Name$ programs converge in $q=(2p+1)n$ steps,
using an absorption replacement process that
we will define.

We start with a simple initial attempt at an absorption replacement process that fails, but provides an intuition about the result. Let $W$ be a walk of length $>(p+1)n$. Then, there exists a vertex $v$ which is visited $t>p+1$ times. Considering the $t-1$ cycles obtained between visits of $v$, let $\bm u$ be the semiring values of the cycles. By applying Lemma~\ref{lemma:Semiring-AbsorptionLemma}, we can create a set of walks $A$, that each visit $v$ at most $p+1$ times by using fewer cycles, such that $A$ absorbs $W$. By repeating this on vertices visited $>p+1$ times, we seemingly could reduce the length of walks to at most $n(p+1)$. Unfortunately, this fails because Lemma~\ref{lemma:Semiring-AbsorptionLemma} can increase the degree of some $u\in \bm u$, which increases the length of a walk by replacing short cycles with long ones. We therefore have to be more careful with how absorption is done.

\subsection{Preliminary Definitions}

We first define a loop tree, a convenient way to organize the loops of a walk into a hierarchy. We will use this hierarchy to repeatedly do absorption to remove loops, in a top-down manner.

\begin{definition}
Let  $W$ be a walk on a directed graph $G=(V, E)$, where $W$ has path-loop decomposition $(P,\mathcal{L})$.
The loop tree $\tau$ for $(P,\mathcal{L})$ is a tree
with the following properties:
\begin{itemize}
    \item 
The vertices of $\tau$ are the path $\{P\}$ and the loops in $\mathcal{L}$. 
\item 
The root of $\tau$ is the path $P$.
\item 
The vertices $\mathcal L_d$ at depth $d \ge 1$ in $\tau$
are the loops $L$ in $\mathcal L$ that intersect some loop (or path) $L'$ of depth $d-1$ in $\tau$, but no loop (or path) of depth less than $d-1$ in $\tau$. The parent of $L$ is an arbitrary such $L'$. $\mathcal L_0$ consists of the path $P$.
\end{itemize}
\end{definition}

Note that the loop tree is not uniquely determined by the path-loop decomposition.
We denote by ${\mathcal L}_d(v)$ the collection of loops in ${\mathcal L}_d$ that contain a vertex $v \in V$.
For some $L \in \mathcal L$, let $\tau(L)$ be the subtree of $\tau$ rooted at  node $L$.

Given a path-loop decomposition, we can often get many different walks, by ordering loops in different ways. Below, we describe a more specific way to create a walk from a loop tree.

\begin{definition}
Let $\tau$ be a loop tree, or a sub-tree of a loop-tree. Let $L'$ be the root of  $\tau$. If $L'$ is a path, let $v$ be its first vertex, and if $L'$ is a loop, let $v$ be any vertex in $L'$. The canonical walk $C(\tau, v)$ of $\tau$ starting at $v$ is recursively defined as follows:
\begin{itemize}
    \item If $\tau$ consists of one node then
    the walk $C(\tau, v)$ consists of traversing $L'$ from $v$ (in the only  way possible).
    \item Otherwise  the walk $C(\tau, v)$ starts traversing $L'$ from $v$. The first time that the walk along $L'$ visits a vertex
    $x$ contained in a child $L$ of $L'$ in $\tau$, where $L$ has not previously been traversed, $C(\tau, v)$ then recursively traverses
    $C(\tau(L), x)$, before returning
    to continue traversing $L'$ from $x$. If for a vertex $x$ there are multiple such children $L$,
    then they can be recursively traversed in arbitrary order.
\end{itemize}
\end{definition}

\begin{definition} For a path-loop decomposition $(P, {\mathcal L})$ on a graph
$G=(V, E)$,
we define the vertex level partitioning of $V$ inductively as follows:
    \begin{enumerate}
        \item $T_0$ is the set of vertices that occur in $P$.
        \item $T_{d}$ is the set of vertices that occur in
         a loop in $\mathcal L_d$ but not in a loop in $\mathcal L_{d-1}$.
    \end{enumerate}
\end{definition}

\begin{observation}
    \label{obs:GapPartition}
    Let  $L$ be a loop in $\mathcal{L}_d$.
   Then it must be the case that $L$ contains at least one vertex in $T_{d-1}$,
    and every vertex in $L$ is either in
    $ T_{d-1}$ or  $ T_{d}$.
\end{observation}

\begin{observation}
    Let $(P, {\mathcal L})$ be a path-loop decomposition. Then for all loops $L \in \mathcal L$
    the depth of $L$ is the same in all possible loop trees of $(P, {\mathcal L})$.
\end{observation}

This shows that the depth of a loop in a loop decomposition is well-defined, without any need to reference a particular loop tree.

An illustrative example walk is shown in Figure  \ref{fig:looptree}.
One possible path-loop decomposition is the path $(s, a, b, t)$, and a loop decomposition
consisting of the loop $(a, c, d)$,  the loop $(b, d, e)$,
the loop $(c, d, f)$, and the loop $(e, g,d)$. One possible loop tree is shown
in Figure  \ref{fig:looptree}. In the vertex partition for this path-loop
decomposition:
$T_0 =\{ s, a, b, t\}$, $T_1 = \{c, d, e\}$, and $T_2 =\{f, g \}$.
If one broke ties in favor of leftmost subtrees,
the canonical walk starting from $s$ for this loop tree would visit
the vertices in the following order:
$s, a, c, d, a, b, d, f, c, d, e, g, d, e, b, t $.


\begin{figure}
  \centering
    \begin{minipage}[b]{0.45\textwidth}
\centering
\scalebox{0.55}{\begin{tikzpicture}
[node distance=1cm,
every node/.style={circle, draw,minimum size=4mm, inner sep=1pt},
->,
>=Stealth,  
semithick,
scale=0.5,
box/.style={draw, rectangle, rounded corners, thick, red, inner sep=5pt}
]

\node (s) {s};
\node[right=of s] (a) {a};
\node[right=of a] (b) {b};
\node[right=of b] (t) {t};

\node[below=of a] (c) {c};
\node[right=of c] (d) {d};
\node[right=of d] (e) {e};

\node[below=of d] (g) {g};
\node[left=of g] (f) {f};

\draw[->] (s) -- (a);
\draw[->] (a) -- (b);
\draw[->] (b) -- (t);

\draw[->] (a) -- (c);
\draw[double,->] (c) -- (d);
\draw[->] (d) -- (a);

\draw[->] (f) -- (c);
\draw[->] (d) -- (f);

\draw[->] (b) -- (d);
\draw[double,->] (d) -- (e);
\draw[->] (e) -- (b);

\draw[->] (e) -- (g);
\draw[->] (g) -- (d);

\end{tikzpicture}}
\caption*{ Walk from vertex $s$ to $t$. Solid edges are traversed once, and double edges are traversed twice.}
\end{minipage}
\hfill
\begin{minipage}[b]{0.45\textwidth}
\centering
\scalebox{0.55}{\begin{tikzpicture}
[node distance=0.5cm,
every node/.style={circle, draw,minimum size=4mm, inner sep=1pt},
->,
>=Stealth,  
semithick,
scale=0.5,
box/.style={draw, rectangle, rounded corners, thick, red, inner sep=2pt}
]

\node (s) {s};
\node[right=of s] (a) {a};
\node[right=of a] (b) {b};
\node[right=of b] (t) {t};

\node[box, fit=(s)(a)(b)(t)] (box1) {};

\node[below left=of a, yshift=-0.1cm] (a2) {a};
\node[below=of a2] (c) {c};
\node[right=of c] (d1) {d};

\node[box, fit=(a2)(c)(d1)] (box2) {};

\node[below right=of b, yshift=-0.1cm] (b2) {b};
\node[below=of b2] (d2) {d};
\node[right=of d2] (e) {e};

\node[box, fit=(b2)(d2)(e)] (box3) {};

\node[below left=of d2, yshift=-0.1cm] (d3) {d};
\node[left=of d3] (c1) {c};
\node[below=of c1] (f) {f};

\node[box, fit=(d3)(c1)(f)] (box4) {};

\draw[->] (d3) -- (f);
\draw[->] (f) -- (c1);
\draw[->] (c1) -- (d3);

\node[below right=of d2, yshift=-0.1cm] (d4) {d};
\node[right=of d4] (e1) {e};
\node[below=of e1] (g) {g};
\node[box, fit=(d4)(e1)(g)] (box5) {};

\draw[->] (d4) -- (e1);
\draw[->] (e1) -- (g);
\draw[->] (g) -- (d4);

\draw[-] (box1) -- (box2);
\draw[-] (box1) -- (box3);
\draw[-] (box3) -- (box4);
\draw[-] (box3) -- (box5);

\draw[->] (s) -- (a);
\draw[->] (a) -- (b);
\draw[->] (b) -- (t);

\draw[->] (a2) -- (c);
\draw[->] (c) -- (d1);
\draw[->] (d1) -- (a2);

\draw[->] (b2) -- (d2);
\draw[->] (d2) -- (e);
\draw[->] (e) -- (b2);

\end{tikzpicture}}
\caption*{Possible loop tree, where the nodes of the tree are the red boxes.}
\end{minipage}
    \caption{Example construction of a loop tree.}
    \label{fig:looptree}
\end{figure}

\subsection{The Linear Absorption Replacement Process}

We are now ready to describe the absorption replacement process that we will use. Recall that we maintain a set of derivation trees, denoted by $A$, which is initialized with a single arbitrary tree. (Note that in the linear case, these trees are simply walks.) 
To aid in the analysis,  we will maintain a path-loop decomposition $(P, \mathcal L)$
of each walk $W \in A$. The path $P$ will be the same for all walks in
$A$, but the loop decomposition might be different for different walks in $A$. Let $s$ be the starting vertex in $P$ and $t$ the terminal vertex in $P$.

Let  $W$ be a walk in $A$, with an associated path-loop decomposition $(P, \mathcal L)$,
and let $\tau$ be a loop tree for $(P, {\mathcal L})$, that contains more than $(2p+1)n$ edges.
We now describe how to construct a set $B$ of walks (not containing $W$), with associated loop decompositions, that collectively absorb
$W$.  We break the construction into cases.

\noindent
{\bf First Case:}
If there is some loop $L$ with multiplicity at least $p+2$
in $\mathcal L$,
then $B$ will consist of $p+1$ walks $W_1, \ldots, W_{p+1}$, where $W_i$ is an arbitrary
walk on the edges of $W$ minus $i$ copies of $L$.

For the second and third cases, we can assume that no loop $L$ has multiplicity greater than $p+1$ in  $\mathcal L$. Let $d$ be the minimum such that there exists a vertex $v$
such that either
\begin{itemize}
    \item $v \in T_{d-1}$ and the cardinality of $\mathcal{L}_{d}(v)$ is at least $p+1$ (second case)
    \item $v \in T_{d}$ and the cardinality of $\mathcal{L}_{d}(v)$ is at least $p+2$ (third case)
\end{itemize}
Note that such a $d$ must exist. Otherwise every vertex would have degree at most $2p + 1$, which contradicts the assumption that the walk has more than $(2p+1)n$ edges.

\noindent
{\bf Second Case:}
In the second case assume that there is a $v \in T_{d-1}$ such that the cardinality of $\mathcal{L}_{d}(v)$ is at least $p+1$.
Let $\tau'$ be equal to $\tau$ with the subtrees rooted at each $L \in {\mathcal L}_{d}(v)$ removed.
Let  $\pi:[1, m] \rightarrow {\mathcal L}_{d}(v)$
be a function that specifies a sequence of $m$ loops $L_{\pi(1)}  \ldots L_{\pi(m)}  $ from ${\mathcal L}_{d}(v)$,
where $m$ is an integer in $[0, p]$. Note there are no other restrictions on $\pi$ other
than that it is a function, so in particular there is no limit on the number of times each
loop in ${\mathcal L}_{d}(v)$ can appear in the sequence.
For each integer $m\in[0, p]$, for each such $\pi$, there is a  walk $W(\pi) \in B$ that is defined as follows: 
 \begin{enumerate}
     \item First, $W(\pi)$ follows the canonical walk $C(\tau', s)$ until the
 vertex $v$ is reached,
 \item then iteratively for each $i \in [1, m]$, follows the canonical walk  $C(\tau(L_{\pi(i)}), v)$,
 \item and then finally follows the rest of the canonical walk $C(\tau', s)$.
 \end{enumerate}
The associated loop decomposition  for $W(\pi)$ is derived from $\mathcal L$ by deleting the loops
in the subtrees of $\tau$ rooted at loops in ${\mathcal L}_{d}(v)$, and adding the
loops in
the subtrees rooted at $L_{\pi(1)}  \ldots L_{\pi(m)}  $. More formally this is equal to
$$({\mathcal L} \multisetminus (\uplus_{L \in {\mathcal L}_{d}(v)} \uplus_{L' \in \tau(L)} \{L'\} ))
\uplus_{i=1}^m \cup_{L' \in \tau(L_{\pi(i)})} \{L'\} $$
In the boundary case that $p=0$, then $B$ is a singleton consisting of the canonical
walk on $\tau'$, with a loop decomposition derived from $\mathcal L$ by deleting
the loops in the subtree rooted at any loop in ${\mathcal L}_{d}(v)$.

\noindent
{\bf Third Case:}
In the third case assume that there is a  $v \in T_{d}$ such
that the cardinality of $\mathcal{L}_{d}(v)$ is at least $p+2$.
Let $L''$ denote an arbitrary loop in ${\mathcal L}_d(v)$.
Let $\tau'$ be equal to $\tau$ with the subtrees rooted at each $L \in {\mathcal L}_{d}(v) \multisetminus \{L''\}$ removed.
Let  $\pi:[1, m] \rightarrow {\mathcal L}_{d}(v) \multisetminus \{L''\}$
be a function that specifies a sequence of $m$ loops $L_{\pi(1)}  \ldots L_{\pi(m)}  $ from ${\mathcal L}_{d}(v) \multisetminus \{L''\}$,
where $m$ is an integer in $[0, p]$.
For each integer $m\in[0, p]$, for each such $\pi$, there is a  walk $W(\pi) \in B$ that is defined as follows: 
\begin{enumerate}
     \item First, $W(\pi)$ follows the canonical walk $C(\tau', s)$ until the
 vertex $v$ is reached,
 \item then iteratively for each $i \in [1, m]$, follows the canonical walk  $C(\tau(L_{\pi(i)}), v)$,
 \item and then finally follows the rest of the canonical walk $C(\tau', s)$.
 \end{enumerate}
The associated loop decomposition  is derived from $\mathcal L$ by deleting the loops
in the subtrees of $\tau$ rooted at loops in ${\mathcal L}_{d}(v) \multisetminus \{L''\}$, and adding
loops in
the subtrees rooted at $L_{\pi(1)}  \ldots L_{\pi(m)}  $. More formally this is equal to
$$({\mathcal L} \multisetminus ( \uplus_{L \in {\mathcal L}_{d}(v) \multisetminus \{L''\}} \uplus_{L' \in \tau(L)} \{L'\} ))
\uplus_{i=1}^m \uplus_{L' \in \tau(L_{\pi(i)})} \{L'\} $$
In the boundary case that $p=0$, then $B$ is a singleton consisting of the canonical
walk on $\tau'$, with a loop decomposition derived from $\mathcal L$ by deleting
the loops in the subtree rooted at any loop in ${\mathcal L}_{d}(v) \multisetminus \{L''\} $.

\begin{lemma}
    Collectively the walks in $B$ absorb $W$, and $W \notin B$.
\end{lemma}
\begin{proof}
We break the proof into three cases, corresponding to the three cases in the construction of $B$ from $W$.
    In the first case, this follows from the $p$-stability of the loop $L$.
    In the second case,
    we can express the yield of $W$ algebraically as
    $W = \tau' \bigotimes_{i=1}^r u_i^{m_i}$
    where the $u_i$'s are the yields of the distinct subtrees $\tau_1, \ldots, \tau_r$ of the form $\tau(L)$ for $L \in \mathcal{L}_{d}(v)$.
    Then the number of distinct walks in $B$ that choose $k_i$ copies of $\tau_i$ is
    $\binom{m}{k_1,\dots,k_r}\bigotimes_{j=1}^ru_j^{k_j}$
    where $m = \sum_{i=1}^r k_i$ and $m \in [0, p]$. Thus $B$ absorbs $W$ by Lemma~\ref{lemma:Semiring-AbsorptionLemma} (Parallel Absorption Lemma).
    
    The proof for the third case is essentially identical to the proof for the second case.
\end{proof}

A subtle point here is that the walks in
$B$ may not be shorter than $W$, so neither termination
nor the fact that $A$ only contains walks shorter
than the original walk is obvious.
However, the construction of $B$ was carefully designed so
that this absorption replacement process will terminate with a set $A$ of shorter walks.
Roughly speaking, this is because
for every walk $W' \in B$, the sequence of cardinalities 
$| {\mathcal L}_1 |, | {\mathcal L}_2 |, \ldots $
is lexicographically smaller than the corresponding sequence
for the walk $W$ from which $B$ was created.
Thus this process must terminate.

The resulting walks $A$ are short because when the absorption replacement process terminates, when viewed as a loop tree, for a walk in $A$, a vertex $v\in T_d$ is visited by at most $p+1$ loops in $\mathcal{L}_d$, and $p$ loops in $\mathcal{L}_{d+1}$. Then each vertex is visited at most $2p+1$ times, giving the desired bound. This therefore establishes Theorem~\ref{thm:linearmain}.

\section{Upper Bounds for General Datalogo Programs}
\label{sec:generalupperbound}

The goal of this section is to give an overview of the proof of Theorem \ref{thm:nonlinearmain},
by defining an absorption replacement process producing a collection $A$
of shallow derivation trees that absorb an arbitrary deep derivation tree $T$,
by using the linear case as a black box.
This process starts by transforming $T$ into what we call a canonical tree $T'$.
We then explain how to construct the
multiset $B$ of derivation trees from a canonical tree by
applying the linear absorption replacement process.
We will use compatible trees throughout the section, as defined in Section~\ref{sec:prelim}, but refer to these as derivation trees.
Due to space limitations, the full proof of Theorem \ref{thm:nonlinearmain} \apref{is moved to Appendix~\ref{sec:generalupperbound:app}}{can be found in the full version of this paper}.

\subsection{ Path Tree}
\label{subsec:leftpathtree}

Consider a derivation tree $T$.  We will now define a new tree $L(T)$, which we call a \textit{path tree}, 
that is derived from $T$. This path tree allows us to view a derivation tree as a
hierarchical collection of paths/walks. This is crucial, as it will allow us to use the linear absorption from the previous section as a black box. Each node in $L(T)$ corresponds to some
ancestor-to-descendant path in $T$, but not all ancestor-to-descendant paths in $T$ will be
in $L(T)$.

\begin{restatable}[Path Tree]{definition}{pathtreedef} $L(T)$ denotes the decomposition, or \textit{path-tree}, of a derivation tree $T$, which is defined as follows: Let $P(T)$ be the left-most root-to-leaf path according to the child ordering. Set the root of $L(T)$ to correspond to $P(T)$. Consider deleting $P(T)$ from $T$ to get a collection of trees $T_1, T_2, \ldots$, where the root of each tree is a child of a node in $P(T)$. Recurse on $T_1, T_2, \ldots$ to get a collection of trees $L(T_1), L(T_2), \ldots$, and set each as a child of the root in $L(T)$.
\end{restatable}

\begin{figure}[h!]
\centering
\definecolor{ptRed}{HTML}{E27D7D}
\definecolor{ptGreen}{HTML}{7FB069}
\definecolor{ptBlue}{HTML}{6FA8DC}
\definecolor{ptYellow}{HTML}{F5C26B}
\begin{minipage}[b]{0.40\textwidth}
\centering
\scalebox{0.65}{\begin{tikzpicture}
[node distance=2.5mm and 1.5mm,
every node/.style={circle, draw,minimum size=6.5mm, inner sep=1pt},semithick]
\node[fill=ptRed] (n0) {1};
\node[fill=ptRed, below left=of n0] (n1) {2};
\node[fill=ptGreen, below right=of n0] (n2) {3};
\node[fill=ptRed, below left=of n1] (n3) {4};
\node[fill=ptGreen, below left=of n2] (n5) {6};
\node[fill=ptYellow, below right=of n2] (n6) {7};
\node[fill=ptBlue, below right=of n3] (n4) {5};
\node[fill=ptRed, below left=of n3] (n7) {8};
\draw (n0) -- (n1);
\draw (n0) -- (n2);
\draw (n1) -- (n3);
\draw (n3) -- (n4);
\draw (n2) -- (n5);
\draw (n2) -- (n6);
\draw (n3) -- (n7);
\end{tikzpicture}}
\end{minipage}
\hspace*{0.1\textwidth}
\begin{minipage}[b]{0.40\textwidth}
\centering
\scalebox{0.65}{\begin{tikzpicture}
[node distance=2.5mm and 1.5mm,
every node/.style={circle, draw,minimum size=6.5mm, inner sep=1pt},semithick]
\node[fill=ptRed, circle] (r) {1,2,4,8};
\node[fill=ptGreen, circle, below left=of r] (g) {3,6};
\node[fill=ptBlue, circle, below right=of r] (b) {5};
\node[fill=ptYellow, circle, below left=of g] (y) {7};
\draw (r) -- (g);
\draw (r) -- (b);
\draw (g) -- (y);
\end{tikzpicture}}
\end{minipage}
\caption{The left figure shows a derivation tree $T$  and the right shows its path tree $L(T)$. This example shows $4$ ancestor-to-descendant paths of $T$ and their $4$ corresponding path-nodes in $L(T)$.
A root-to-leaf path of $T$ can be thought of as a path from the root to some path-node in the decomposition $L(T)$, where only a prefix of the path of every path-node is traversed.
Here, $P(T) = (1,2,4,8)$.
}
\label{fig:tree-decomposition}
\end{figure}

Each node in $L(T)$ corresponds to a node-to-leaf path in $T$. To avoid confusion between
derivation trees and path trees, we will call a node in a derivation tree $T$ a
\textit{node} and a node in a decomposition $L(T)$ a \textit{path-node}. 
Note that a path-node corresponds to multiple nodes in the derivation
tree, which form the left-most path of some sub-tree of the derivation tree. 
When we write path-node, the reader should think of this corresponding path in $T$.
Figure~\ref{fig:tree-decomposition} illustrates this decomposition.

\subsection{Canonical Tree Transformation}
\label{sec:transform}

Here we will describe a procedure to transform any derivation tree into what we call a canonical tree. This transformation reorders different subtrees such that the tree becomes as ``left-deep'' as possible. This is useful, because the linear absorption procedure that we will use as a black-box works on paths/walks, so it is more effective on left-deep trees.

Before we state the canonical tree construction process, we need some additional definitions.

\begin{restatable}[Wedge]{definition}{wedgedef}
A wedge $W$ of a derivation tree $T$ is defined by an ancestor-to-descendant path $x, y_1,
\ldots, y_k, x$ between two nodes representing the same nonterminal $x$ from the underlying
context-free grammar $G$. Then $W$ is the subtree of $T$ rooted at the first $x$ with the subtree rooted at the second $x$ removed. The root-leaf path $x,y_1,\dots,y_k$ in $W$ is called $W$'s spine.
\end{restatable}

\begin{restatable}[Movable Wedge]{definition}{movablewedgedef}
Consider a derivation tree $T$. A wedge $W$ is considered \textit{movable} with respect to a
path-node $P$ (in $L(T)$) and a node $v$, if the nodes forming $W$'s spine
appear within descendant path-nodes of $P$, and both $P$ and $W$'s spine include $v$.
\end{restatable}

Intuitively, wedges generalize loops. We can move a movable wedge to a higher left-most path using the following \textit{wedge push}. The operation is illustrated in Figure~\ref{fig:wedgepush}.

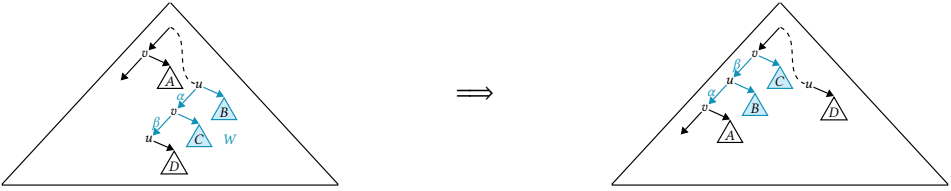
\begin{figure}[H]
\def\xoff{-0.6}
\def\yoff{-0.65}
\begin{minipage}[c]{0.42\textwidth}
\centering
\scalebox{0.55}{
\begin{tikzpicture}[node distance=0.5cm,
every node/.style={minimum size=0mm, inner sep=0pt},
->,>=Triangle,thick]
\tikzset{
trianglenode/.style={
regular polygon,
regular polygon sides=3,
minimum size=0.7cm,
inner sep=0pt,
text centered
},
wedge/.style={draw=cyan!70!black,fill=cyan!18},
spine/.style={draw=cyan!70!black}
}

\node (n0) {};
\node (n1) at ($(n0)+(\xoff,\yoff)$) {$v$};
\node(n2) at ($(n1)+(\xoff,\yoff)$) {};
\node[trianglenode,draw] (n3) at ($(n1)+(-\xoff,\yoff)$) {$A$};

\draw(n0) -> (n1);
\draw (n1) -> (n2);
\draw (n1) -> (n3.north);

\node (m0) at ($(n0) + (0.7,-1.4)$) {$u$};
\node[trianglenode,wedge] (tB) at ($(m0)+(-\xoff,\yoff)$) {$B$};
\node (m1) at ($(m0)+(\xoff,\yoff)$) {$v$};
\node[trianglenode,wedge] (tC) at ($(m1)+(-\xoff,\yoff)$) {$C$};
\node (m2) at ($(m1)+(\xoff,\yoff)$) {$u$};
\node[trianglenode,draw] (tD) at ($(m2)+(-\xoff,\yoff)$) {$D$};
\node (wedge) at ($(tC)+(0.75,0)$) {{\color{cyan!70!black} $W$}};

\draw[spine] (m0) -> node[pos=0.4,left, xshift=-2pt,text=cyan!70!black] {$\alpha$} (m1);
\draw[spine] (m1) -> node[pos=0.4,left, xshift=-2pt,text=cyan!70!black] {$\beta$} (m2);
\draw[spine] (m0) -> (tB.north);
\draw[spine] (m1) -> (tC.north);
\draw (m2) -> (tD.north);
\draw[dashed,-] (n0) to[out=-25, in=160] (m0);

\node(A) at ($(n0)+(0,0.6)$)  {};
\node(B) at ($(A)+(-4.06,-4.4)$) {};
\node (C) at ($(A)+(4.06,-4.4)$) {};
\draw[thick,-] (A) -- (B) -- (C) -- (A);
\end{tikzpicture}
}
\end{minipage}%
\hfill
\begin{minipage}[c]{0.06\textwidth}
\centering
$\Longrightarrow$
\end{minipage}%
\hfill
\begin{minipage}[c]{0.42\textwidth}
\centering
\scalebox{0.55}{
\begin{tikzpicture}[node distance=0.5cm,
every node/.style={minimum size=0mm, inner sep=0pt},
->,>=Triangle,thick]
\tikzset{
trianglenode/.style={
regular polygon,
regular polygon sides=3,
minimum size=0.7cm,
inner sep=0pt,
text centered
},
wedge/.style={draw=cyan!70!black,fill=cyan!18},
spine/.style={draw=cyan!70!black}
}

\node (n0) {};
\node (n1) at ($(n0)+(\xoff,\yoff)$) {$v$};
\node[trianglenode,wedge] (tC) at ($(n1)+(-\xoff,\yoff)$) {$C$};
\node (n2) at ($(n1)+(\xoff,\yoff)$) {$u$};
\node[trianglenode,wedge] (tB) at ($(n2)+(-\xoff,\yoff)$) {$B$};
\node (n3) at ($(n2)+(\xoff,\yoff)$) {$v$};
\node[trianglenode,draw] (tA) at ($(n3)+(-\xoff,\yoff)$) {$A$};
\node(n4) at ($(n3)+(\xoff,\yoff)$) {};

\draw(n0) -> (n1);
\draw[spine] (n1) -> node[pos=0.4,left, xshift=-2pt,text=cyan!70!black] {$\beta$} (n2);
\draw[spine] (n2) -> node[pos=0.4,left, xshift=-2pt,text=cyan!70!black] {$\alpha$} (n3);
\draw (n3) -> (n4);

\node (m0) at ($(n0) + (0.7,-1.4)$) {$u$};
\node[trianglenode,draw] (tD) at ($(m0)+(-\xoff,\yoff)$) {$D$};

\draw[spine] (n1) -> (tC.north);
\draw[spine] (n2) -> (tB.north);
\draw (n3) -> (tA.north);
\draw[dashed,-] (n0) to[out=-25, in=160] (m0);
\draw (m0) -> (tD.north);

\node(A) at ($(n0)+(0,0.6)$)  {};
\node(B) at ($(A)+(-4.06,-4.4)$) {};
\node (C) at ($(A)+(4.06,-4.4)$) {};
\draw[thick,-] (A) -- (B) -- (C) -- (A);

\end{tikzpicture}
}
\end{minipage}
\caption{Illustration of a wedge push, where a wedge whose spine consists of subpaths $\alpha$ and $\beta$ is pushed to the left black path (or path-node) $P$. The triangles represent subtrees hanging off the path or the spine. Note that the spine is a descendant of $P$ and may intersect multiple path-nodes.}
\label{fig:wedgepush}
\end{figure}

\begin{restatable}[Wedge Push]{definition}{wedgepushdef}
Given a derivation tree $T$, consider a movable wedge $W$ with respect to a path-node $P \in
L(T)$ and a node $v$. Suppose $W$'s spine begins with $u$. A wedge push of $W$ into
$P$ is defined as follows: If $W$'s spine consists of two subpaths $\alpha$ and $\beta$
split by $v$, then an occurrence of $v$ on $P$ is replaced with a subpath starting with
$\beta$ and ending with $\alpha$. Here, all subtrees hanging off from the spine also move to
$P$, except the one hanging off the spine on its bottom end.
\end{restatable}

Note that this is still well defined if $v = u$; then we do not change the permutation of the spine. A wedge push moves around parts of the tree such that the semiring value remains the same, while the tree itself becomes more left-deep. 

Next, we define the {\bf Canonical Tree Transformation} procedure: it takes as input a derivation tree $T$, and then while there exists a path-node $P \in L(T)$ with a movable wedge $W$ w.r.t. node $v$, it performs a wedge-push on $W$ to $P$ w.r.t. $v$.


\begin{restatable}{lemma}{cantreetransformterminate}
The Canonical Tree Transformation Procedure terminates.
\end{restatable}


The proof follows by showing that each wedge-push operation strictly increases a vector in lexicographic order. Here, the $i$-th coordinate corresponds to the total number of nodes in $T$ that belong to path-nodes at level $i$ in $L(T)$. We also note that the transformation procedure does not modify the yield of the tree.

The following lemma states some structural properties of a canonical tree. The intuition here is that if these properties do not hold, we can do another wedge push.

\begin{restatable}{lemma}{PostTransProcedureProps}
\label{lemma:PostTransProcedureProps}
Let $T$ be a canonical tree produced by  the canonical tree transformation procedure.
Consider a root to leaf path $\rho$ of path-nodes in the path tree $L(T)$. Fix a node $v$ in
$T$ that appears in some path-node of $\rho$. Let $P_1, P_2, \ldots$ be the path-nodes in
$\rho$ containing $v$, ordered by increasing depth in $L(T)$. We have the
following.
\begin{enumerate}
\item Every path-node $P_i$ where $i\geq 2$ has at most $1$ copy of $v$ in the path
corresponding to the node in $T$.  The path corresponding to $P_1$ can have more than one
copy of $v$.
\item Fix any path-node $P_i$ where $i \geq 2$.  Consider the path in $T$ corresponding to
$P_i$. The set of nodes before $v$ in the path is disjoint from the
set of nodes occurring after $v$ in the path.
\end{enumerate}
\end{restatable}

\subsection{Application of the Linear Absorption Replacement Process to Derivation Trees}
\label{subsec:reductiontolinear}

Given our assumptions about the structure of derivation trees, our linear
absorption replacement process $\textsf{LinAbsorb}(P)$ can be applied to any
ancestor-to-descendant path $P$ of a derivation tree $T$. This creates a collection $A$ of
paths, derived from $P$, which absorb $P$. We will explain how to use
$A$ and $T$ to create a collection of trees $A'$ that absorb $T$.
Let $(X, Y)$ be a directed edge in $P$, where $X$ is the parent of $Y$ in $T$.
Let $T_{X, 1}, \ldots, T_{X, k}$ be the subtrees of the children of $X$ in $T$, other than
$Y$. We will think of the subtrees  $T_{X, 1}, \ldots, T_{X, k}$ as being
associated with the edge $(X, Y)$ in $T$; when \textsf{LinAbsorb} duplicates $(X, Y)$, the
subtrees  $T_{X, 1}, \ldots, T_{X, k}$ are also duplicated, with each copy of $X$
getting duplicate copies of the subtrees as children. Similarly when \textsf{LinAbsorb} deletes $(X, Y)$ the
subtrees $T_{X, 1}, \ldots, T_{X, k}$  are deleted.
Note for completeness that the trees $A'$ produced are all compatible. 
See Figure~\ref{fig:path-subtrees} for
an illustration, where for simplicity we assume $k$ is universally 1.

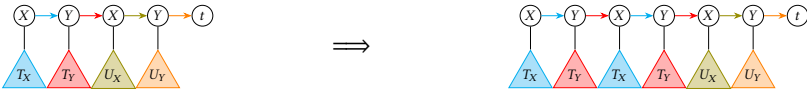
\begin{figure}[H]
\begin{minipage}{\textwidth}
\centering
\begin{minipage}[c]{0.41\textwidth}
\centering
\scalebox{0.55}{
\begin{tikzpicture}
[node distance=0.55cm,
        every node/.style={circle, draw,minimum size=5mm, inner sep=1pt},
        ->,
        >=Stealth,  
        semithick,
        scale=0.5,
        box/.style={draw, rectangle, rounded corners, thick, red, inner sep=5pt}
        ]
     \tikzset{
    trianglenode/.style={
      regular polygon,
      regular polygon sides=3,
      minimum size=1.2cm,
      inner sep=0pt,
      text centered
    }}
\node (a) {$X$};
\node[right=of a] (b) {$Y$};
\node[right=of b] (c) {$X$};
\node[right=of c] (d) {$Y$};
\node[right=of d] (e) {$t$};

\node[trianglenode,draw=cyan, fill=cyan!30,below=of a] (ta) {$T_X$};
\node[trianglenode,draw=red, fill=red!30,below=of b] (tb) {$T_Y$};
\node[trianglenode,draw=olive, fill=olive!30,below=of c] (tc) {$U_X$};
\node[trianglenode,draw=orange, fill=orange!30,below=of d] (td) {$U_Y$};

\draw[cyan,->] (a) -- (b);
\draw[red,->] (b) -- (c);
\draw[olive,->] (c) -- (d);
\draw[orange,->] (d) -- (e);

\draw[-] (a) -- (ta);
\draw[-] (b) -- (tb);
\draw[-] (c) -- (tc);
\draw[-] (d) -- (td);

\end{tikzpicture}
}
\end{minipage}%
\hfill
\begin{minipage}[c]{0.05\textwidth}
\centering
$\Longrightarrow$
\end{minipage}%
\hfill
\begin{minipage}[c]{0.53\textwidth}
\centering
\scalebox{0.55}{
\begin{tikzpicture}
[node distance=0.55cm,
        every node/.style={circle, draw,minimum size=5mm, inner sep=1pt},
        ->,
        >=Stealth,  
        semithick,
        scale=0.5,
        box/.style={draw, rectangle, rounded corners, thick, red, inner sep=5pt}
        ]
     \tikzset{
    trianglenode/.style={
      regular polygon,
      regular polygon sides=3,
      minimum size=1.2cm,
      inner sep=0pt,
      text centered
    }}
\node (a) {$X$};
\node[right=of a] (b) {$Y$};
\node[right=of b] (c) {$X$};
\node[right=of c] (d) {$Y$};
\node[right=of d] (e) {$X$};
\node[right=of e] (f) {$Y$};
\node[right=of f] (g) {$t$};

\node[trianglenode,draw=cyan, fill=cyan!30,below=of a] (ta) {$T_X$};
\node[trianglenode,draw=red, fill=red!30,below=of b] (tb) {$T_Y$};
\node[trianglenode,draw=cyan, fill=cyan!30,below=of c] (tc) {$T_X$};
\node[trianglenode,draw=red, fill=red!30,below=of d] (td) {$T_Y$};
\node[trianglenode,draw=olive, fill=olive!30,below=of e] (te) {$U_X$};
\node[trianglenode,draw=orange, fill=orange!30,below=of f] (tf) {$U_Y$};

\draw[cyan,->] (a) -- (b);
\draw[red,->] (b) -- (c);
\draw[cyan,->] (c) -- (d);
\draw[red,->] (d) -- (e);
\draw[olive,->] (e) -- (f);
\draw[orange,->] (f) -- (g);

\draw[-] (a) -- (ta);
\draw[-] (b) -- (tb);
\draw[-] (c) -- (tc);
\draw[-] (d) -- (td);
\draw[-] (e) -- (te);
\draw[-] (f) -- (tf);

\end{tikzpicture}
}
\end{minipage}
\captionof{figure}{Illustration of duplication (and possible removal) of subtrees when \textsf{LinAbsorb} is applied to a path. The left shows a node-to-leaf path $P$ in $T$ and subtrees, shown as triangles.
Note that $P$ ends with the sink $t$.
While $T_X$ is a subtree of $X$, we treat it as being associated with edge $(X,Y)$. Note that two copies of the same edge can have different subtrees. The right shows a new path/walk created by \textsf{LinAbsorb} along with duplicated subtrees.
}
\label{fig:path-subtrees}
\end{minipage}
\end{figure}

\subsection{Absorption Replacement Process}
\label{subsec:nonlinearreplacementprocess}

We here describe one step of our absorption replacement process (Algorithm~\ref{alg:absorb}), which produces a collection $B$ of trees from a tree $T$. 
Each step applies the linear absorption replacement process to a path-node in $L(T)$. This process makes the derivation trees shallower, and  shallower trees correspond to shorter proofs for the same output tuple. In Algorithm~\ref{alg:absorb}, $\text{length}(P)$ is the number of edges in $P$.

\subsection{Analysis of Replacement Process}
\label{sec:nonlinearanalysis}

\begin{wrapfigure}[9]{r}{0.50\textwidth}
\vspace{-2.5\baselineskip}
\begin{minipage}{0.50\textwidth}
\small
\captionof{algorithm}{Absorption Replacement}\label{alg:absorb}
\begin{algorithmic}[1]
    \REQUIRE Input derivation tree $T$
    \STATE \textbf{Procedure} Absorb($T$)
    \STATE $T \gets$ Canonical Transformation of $T$
    \STATE \textbf{if } $\exists P \in L(T)$ with $\text{length}(P) > n(2p+1)$
    \STATE \hspace{1em} $P \gets$ a lowest-depth such path in $L(T)$
    \STATE \hspace{1em} $P_1,\dots,P_k \gets \text{\textsf{LinAbsorb}}(P)$
    \STATE \hspace{1em} $T_1,\dots,T_k \gets$ trees from $T$ using $P_1,\dots,P_k$
    \STATE \hspace{1em} $B \gets \{T_1,\dots,T_k\}$
    \STATE \textbf{return} $B$
\end{algorithmic}
\end{minipage}
\end{wrapfigure}

To establish Theorem~\ref{thm:nonlinearmain}, we need to prove that the absorption replacement process terminates, and that the trees in the resulting set are not too deep.

\begin{restatable}{lemma}{GeneralAbsReplacementTermination}
The absorption replacement process terminates.
\end{restatable}
\begin{proof}
The process recursively applies \textsf{LinAbsorb} on path-nodes in order of increasing depth. Each call only modifies subtrees hanging off the input path (without increasing the depth of $L(T)$) and outputs a path of length $\leq (2p+1)n$.
\end{proof}

While Lemma~\ref{lemma:PostTransProcedureProps} describes canonical trees \textit{before} the absorption replacement process, the next lemma describes their structure \textit{after}, helping us bound the resulting tree depth. For intuition, the reader can observe the similarities between the properties listed in these two lemmas.

\begin{restatable}{lemma}{KeyPropertiesAfterAbsorptionGeneral}
\label{lem:non-linear-key}
Let $A$ be the collection of derivation trees obtained after the absorption replacement process (which started with $A = \{T_0\}$). Let $T \in A$ and let $Q$ be a root-leaf path in $T$.
We denote by $\beta(Q)$ the unique path in the path tree $L(T)$ corresponding to $Q$ (this goes from the root path-node of $L(T)$ to some other path-node in $L(T)$).
For a path-node $P \in \beta(Q)$, we will denote by $s(P)$ the maximal subpath (prefix) of $P$ that intersects with $Q$.
Let $v$ be an arbitrary node in $T$. Then:
\begin{enumerate}
\item If there are path-nodes $P_1$ and $P_2$ in $\beta(Q)$, where $P_1$ is above $P_2$ in $\beta(Q)$, such that $v$ appears in both the walk $P_1$ and prefix $s(P_2)$, then (i) $P_2$ has exactly one occurrence of $v$; (ii) all nodes before $v$ in $P_2$ are disjoint from all nodes after $v$ in $P_2$; and (iii) there is no third descendant path-node $P_3$ of $P_2$ in $\beta(Q)$ such that $P_3$ includes $v$.
\item The first path-node $P$ in $\beta(Q)$ containing $v$ (if it exists) has at most $2p+1$ appearances of $v$.
\end{enumerate}
\end{restatable}

\begin{figure}[H]
\centering
\begin{minipage}[c]{0.62\textwidth}
 
We show an example of properties given by the above lemma.
Here we show an example derivation tree $T$, which consists of $3$ path-nodes $P_1, P_2, P_3$. A root-leaf path $Q$ is highlighted in blue, and in this case $\beta(Q)$ consists of all three path-nodes. $s(P_1), s(P_2), s(P_3)$ are the prefixes of the path-nodes which use the dotted lines. In this example, $(1)(iii)$ is not satisfied since $P_3$ includes $v$.

\end{minipage}\hspace{0.01\textwidth}%
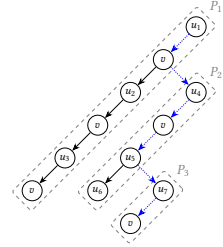
\begin{minipage}[c]{0.34\textwidth}
\centering
\scalebox{0.48}{\begin{tikzpicture}[node distance=5mm and 5mm,
    every node/.style={circle, draw, minimum size=5.5mm, inner sep=1pt, font=\small},
    ->, >=Stealth, semithick]
    \node (n0) {$u_1$};
    \node[below left=of n0] (n1) {$v$};
    \node[below left=of n1] (n2) {$u_2$};
    \node[below left=of n2] (n3) {$v$};
    \node[below left=of n3] (n4) {$u_3$};
    \node[below left=of n4] (n5) {$v$};
    \node[below right=of n1] (n6) {$u_4$};
    \node[below left=of n6] (n7) {$v$};
    \node[below left=of n7] (n8) {$u_5$};
    \node[below left=of n8] (n9) {$u_6$};
    \node[below right=of n8] (n10) {$u_7$};
    \node[below left=of n10] (n11) {$v$};
    \begin{scope}[on background layer]
      \node[draw=gray, dashed, thick, rectangle, rounded corners=2pt, rotate fit=-45, fit={(n0)(n5)}, inner xsep=5pt, inner ysep=6pt] (box1) {};
      \node[draw=gray, dashed, thick, rectangle, rounded corners=2pt, rotate fit=-45, fit={(n6)(n9)}, inner xsep=5pt, inner ysep=6pt] (box2) {};
      \node[draw=gray, dashed, thick, rectangle, rounded corners=2pt, rotate fit=-45, fit={(n10)(n11)}, inner xsep=5pt, inner ysep=6pt] (box3) {};
    \end{scope}
    \node[draw=none, font=\large, gray] at ([xshift=10pt,yshift=10pt]n0.north east) {$P_1$};
    \node[draw=none, font=\large, gray] at ([xshift=10pt,yshift=10pt]n6.north east) {$P_2$};
    \node[draw=none, font=\large, gray] at ([xshift=10pt,yshift=10pt]n10.north east) {$P_3$};
    \draw[blue,thick,densely dotted] (n0) -- (n1);
    \draw (n1) -- (n2); \draw (n2) -- (n3); \draw (n3) -- (n4); \draw (n4) -- (n5);
    \draw[blue,thick,densely dotted] (n1) -- (n6);
    \draw[blue,thick,densely dotted] (n6) -- (n7);
    \draw[blue,thick,densely dotted] (n7) -- (n8);
    \draw (n8) -- (n9);
    \draw[blue,thick,densely dotted] (n8) -- (n10);
    \draw[blue,thick,densely dotted] (n10) -- (n11);
\end{tikzpicture}}
\end{minipage}
\label{fig:non-linear-key-illustration}
\end{figure}

The final lemma establishes that our method of absorbing along paths obtains absorbing sets of trees which are shallow, which will imply convergence as we want.

\begin{lemma}
\label{lem:nonlinear-final-in-t}
When the absorption replacement process terminates, any resulting derivation tree has depth at most $n(2p+2)$.
\end{lemma}
\begin{proof}
Consider any such tree $T$ and any root-to-leaf path $Q$ in $T$. Let $\beta$ be the path in
$L(T)$ corresponding to $Q$. Because of (iii) of Property (1) of
Lemma~\ref{lem:non-linear-key}, there are at most two path-nodes in $\beta$ including $v$.
Further, the first and second path-nodes (if they exist) can have at most $(2p+1)$ and $1$
occurrences of $v$, respectively. Thus, $Q$ can have at most $2p+2$ occurrences of $v$.
\end{proof}

The lemma above gives that any tree in the absorbing set $A$ has depth at most $(2p+2)n$, which establishes a convergence time of $2(p+1)n$ for a naive evaluation of an arbitrary
$\Name$ instance, proving Theorem~\ref{thm:nonlinearmain}.

\section{Lower Bounds}
\label{sec:lowerbounds}

In this section, we show lower bounds on a class of algorithms that compute the solution of
a fixpoint of linear equations over a semiring. We begin by introducing relaxation algorithms~\cite{CLRS}, the class of algorithms we establish lower bounds on. Then, we show in Section~\ref{sec:warmup} the
lower bound for a class of instances that are $0$-stable, and the computation graph is a complete graph, which highlights the core technique. Finally, we generalize to $p$-stable semirings and sparse instances.

\subsection{Relaxation Algorithms}
Consider the following generic form of a system of linear polynomial equations:
\begin{equation} \label{eq:lineq}
    \mathbf{X} = \mathcal{A}^\top \otimes \mathbf{X} \oplus \mathbf{b}
\end{equation}
where $\mathbf{X}$ is an $n$-dimensional vector of variables, $\mathcal{A}$ is an $n  \times
n$ matrix that takes values from the domain of $\mathbf{S}$, and $\mathbf{b}$ is an
$n$-dimensional vector with values from the same domain. This system can be thought of as the
grounding of the following linear Datalog program
\begin{align*}
    & R(x) \dl B(x) && & R(x) \dl R(y), E(y,x)
\end{align*}
where $E$ can be viewed as a representation of the edges of a directed graph
$G_{\mathcal{A}}$ with vertices in $[n]$. Recall that our goal is to compute the fixpoint of this system. We denote by $m$ the number of non-zero entries in the matrix $\mathcal{A}$ --
the number of edges in $G_{\mathcal{A}}$.

We first analyze the semiring complexity of naïve evaluation. Using our upper bound on the number of iterations ($(2p+1)n$ for linear systems), and the fact that every iteration requires $m$ semiring multiplications and at most $m$
semiring additions, we obtain the following upper bound.

\begin{theorem}
The fixpoint of Eq.~\eqref{eq:lineq} over a $p$-stable semiring $\mathbf{S}$ can be computed via naïve evaluation with $2(2p+1)mn$ semiring operations.
\end{theorem}

It is known that to find a fixpoint over the Boolean semiring ($0$-stable), we need $O(m)$ semiring operations~\cite{GroundingBasedApproach}, and for the tropical semiring (also $0$-stable), we need $O(m \log n)$ operations~\cite{GroundingBasedApproach}. There is no known algorithm that works for any $0$-stable semiring with $O(m)$ semiring complexity.

{\em Relaxation algorithms} are algorithms that evaluate the fixpoint of
Eq.~\eqref{eq:lineq} over a given semiring $\mathbf{S}$. A relaxation algorithm starts with
initializing the value of each vertex $i \in [n]$ to $\mathbf{X}_i \gets b_i$.  At
each relaxation step, it chooses an edge $(i,j)$ in $G_{\mathcal{A}}$, and updates the value of
$\mathbf{X}_j$ using the newest value of $\mathbf{X}_i$ (and keeping the values of all other
variables the same). If the semiring is 0-stable, then a relaxation via the edge $(i,j)$ can be implemented as the operation $\mathbf{X}_j \gets \mathbf{X}_j \oplus (\mathcal{A}_{ij} \otimes \mathbf{X}_i)$, which needs two semiring operations (this is because 0-stability implies $\oplus$-idempotence). For a general $p$-stable semiring, a relaxation step needs $O(\log n)$ semiring operations\footnote{Since we need to update the value of a sum of size $O(n)$, which can be achieved by a binary tree structure.}.
We give further motivation behind relaxation algorithms in \apref{Appendix~\ref{sec:lowerbound:app}}{the full version of this paper}.

The {\em edge schedule} of an execution of a relaxation algorithm is the sequence of edges that traces the relaxation steps of the algorithm. The output of an edge schedule on an instance is what the corresponding execution of the relaxation algorithm would return. We say that an edge schedule is {\em correct} if the resulting algorithm produces the correct output. To show a lower bound on the semiring complexity of relaxation algorithms, it suffices to show a lower bound on the length of a correct edge schedule.  We first give some definitions that will be important for the lower bounds.

\begin{definition}[Edge Schedule Cover]
An edge schedule $\sigma$ {\em covers} a derivation tree (walk) $W$ of an IDB fact if $W$ is a subsequence of $\sigma$.
\end{definition}

\begin{lemma} \label{lemma:ValueSumTrees} Let $\sigma$ be the edge schedule of a relaxation
algorithm. Then, for any $i \in [n]$, the value of  $\mathbf{X}_i$ is the sum of the yields
of all walks ending at $i$ and covered by $\sigma$.
\end{lemma}
\begin{proof}
We can show this by induction on the length of the edge schedule. For the base case of an
empty schedule, this is true because only length $0$ walks are covered and hence the sum is
$b_i$.

For the induction step, assume that we have already run schedule $\sigma'$ and we see a new edge
$(i,j)$ in $G_{\mathcal{A}}$. By the induction hypothesis, for any $k \in [n]$, the value of
$\mathbf{X}_k$ before the update is the sum of the yields of all walks of $k$ covered by
$\sigma'$. Observe that the addition of $(i,j)$ to the schedule only changes the set of walks of
$j$ covered by $\sigma$. Hence, it suffices to look at how $\mathbf{X}_j$ is updated. Indeed,
when $(i,j)$ is added to the schedule, then to obtain the new set of covered walks it
suffices to consider the subset of walks of $j$ that consist of walks of $i$ in $\sigma'$ plus
the edge $(i,j)$. Thus, the new value will correspond exactly to updating $\mathbf{X}_j$
with the newest value of $\mathbf{X}_i$.
\end{proof}

This is key for our lower bounds: we will see that relaxation algorithms producing shorter schedules than the lower bound gives sums of yields of covered walks not giving the correct result.

\subsection{Cliques with 0-stable Semirings}
\label{sec:warmup}

We give here a lower bound on the length of a correct edge schedule when
$G_{\mathcal{A}}$ is a (directed) clique with $n$ nodes and the semiring is 0-stable. This argument highlights the core ideas used in the general lower bound.

\subsubsection{Intuition of Lower Bound}
We embed our hard instance in a \emph{color semiring} attached to the
complete graph. Each simple path $P$ visiting all $n$ nodes is given a unique color $c(P)$, and every
edge on $P$ is annotated with that color. A relaxation algorithm can only compute the color $c(P)$
at the destination vertex if $\otimes$ is applied for every edge along $P$, in the
order that the edges appear along the path.
The schedule must then be long enough to cover all $n!$ Hamiltonian paths in the graph.

The problem is then reduced to a combinatorial covering
problem: \emph{How long must a sequence of edges be so that every Hamiltonian path occurs
as a subsequence?} Using a counting argument (Lemma \ref{lemma:AbsorptiveN3LB}) inspired by~\cite{Arratia1999} we show that
any sequence achieving full coverage has length $\Omega(n^{3})$, yielding an $\Omega(n^{3})$ lower bound on the number of semiring operations performed by \emph{any} relaxation algorithm. Since naïve evaluation performs $\Theta(n^{3})$ operations on this instance, the bound is tight.

\subsubsection{Lower Bound Analysis}

We will consider the semiring $\mathbf{C}_k = (\mathcal{P}([k]), \cup, \cap, \{\}, [k])$, which we call the {\em color semiring}. This semiring annotates each edge with a subset of the $k$ colors (represented as integers). Reachability over this semiring (say from node $u$ to $v$) computes the set of colors $c$ for which there is a path from $u$ to $v$ such that every edge in the path contains $c$. The color semiring is 0-stable -- in addition, it is also $\otimes$-idempotent.

We construct an \emph{input instance} $I_n$ for some $n \geq 3$ as follows. Let $\mathcal{P}_n$ be the set of all simple paths that visit all nodes in a directed clique with $n$ nodes (every such path has length $n-1$). Note that $|\mathcal{P}_n| = n!$. We take the semiring $\mathbf{C}_k$ with $k=|\mathcal{P}_n|=n!$ colors, and we assign a unique color $c(P)$ to path $P \in \mathcal{P}_n$. Then, let $\mathcal{A}_{ij} = \{c(P) \mid P \in \mathcal{P}_n: (i,j) \in P\}$ and $\mathbf{b}_i = \{c(P) \mid P \in \mathcal{P}_n: P \text{ starts at } i \}$.
In other words, for each simple path $P$, we create a new unique color and add this color to each edge along the path, and to the unique node that $P$ starts at. An edge schedule must return all colors corresponding to paths in  $\mathcal{P}_n$ to be correct.

\begin{lemma} \label{lemma:clique:cover}
If an edge schedule $\sigma$ is correct for the input instance $I_n$ over the semiring $\mathbf{
C}_{n!}$, then every path in the set $\mathcal{P}_n$ is covered by $\sigma$.
\end{lemma}

\begin{proof}
Suppose there is a simple path $P \in \mathcal{P}_n$ from node $u$ to node $v$ that is not covered by the schedule $\sigma$. We will show that the set of walks ending at $v$ and covered by $\sigma$ does not compute the correct value for $\mathbf{X}_v$, which implies the result from Lemma~\ref{lemma:ValueSumTrees}. The key observation is that the unique color $c(P)$ of this path must be included in the output value of  $\mathbf{X}_v$, and for this to happen it must be that $P$ is included in the set. Indeed, each edge $(x,y)$ in the schedule “propagates” a color from node $x$ to $y$, but only if the color occurred at $x$ at the start of the propagation. To propagate the color from the start to the end of a path, we therefore have to follow the order of the edges of the path. Thus, to obtain $c(P)$ we must have a walk that uses only the edges in $P$, and any such walk must start at vertex $u$.
\end{proof}

\begin{lemma}
\label{lemma:AbsorptiveN3LB}
    Let $\theta\in(0,1]$ be some constant. Then, every edge schedule $\sigma$ that covers at least a $\theta$-fraction of the paths in $\mathcal{P}_n$ satisfies $|\sigma| > \theta^{\frac{2}{n}}\frac{n^3}{3e^3}$.
\end{lemma}

\begin{proof}
    
    First, observe: $|\mathcal{P}_n| = n!>(n/e)^n$. The inequality is obtained from Stirling's approximation. We next bound the number of paths from $\mathcal{P}_n$ which are covered by $\sigma$. Note next that a length $n$ path in a graph can be fully defined by $\lfloor n/2 \rfloor$ edges, since by listing every other edge, we fully describe the sequences of vertices of the path. A sequence of edges of size $|\sigma|$ has only $\binom{|\sigma|}{\lfloor n/2 \rfloor}$ subsequences of length $\lfloor n/2 \rfloor$, and hence it covers at most that many paths.
    Thus, we must have
    $
    \binom{|\sigma|}{\lfloor n/2 \rfloor} \geq \theta |\mathcal{P}_n| > \theta \left(\frac{n}{e}\right)^n
    $
    , if we hope to cover a $\theta$-fraction of paths.
    
    We can bound the left-hand side of this inequality as:
       $ \binom{|\sigma|}{\lfloor n/2 \rfloor}
        \leq \left( \frac{e|\sigma|}{\lfloor n/2 \rfloor} \right)^{\lfloor n/2 \rfloor}
        \leq \left( \frac{3e|\sigma|}{n} \right)^{n/2}$.
    By plugging in our first observation, we finally obtain the desired inequality.
\end{proof}

\begin{theorem}
Any correct edge schedule for the instance $I_n$ over the semiring $\mathbf{C}_{n!}$ has length $> \frac{n^3}{3e^3}$.
\end{theorem}

\begin{proof}
This is a direct corollary of the above two lemmas by setting $\theta=1$.
\end{proof}

\subsection{General Lower Bounds}

Our next result generalizes the prior lower bound to graphs with an arbitrary number of edges $m$. To do this, instead of using a clique graph, we use a family of instances where the transition graph resembles a fractal, consisting of cliques of subgraphs which themselves are cliques. We describe the construction in detail in \apref{the appendix}{the full version of this paper}, including with examples.

\begin{restatable}{lemma}{SparseLowerBound}
    For any constant integer $\alpha \geq 1$, there is a family of instances with $\Theta(n)$ nodes and $m=\Theta(n^{1+1/\alpha})$ edges such that any correct edge schedule over the color semiring has length $\Omega(m \cdot n)$.
\end{restatable}

Finally, we lift the lower bound to $p$-stable semirings. The construction mirrors the $0$-stable setting, with one crucial addition: every Hamiltonian path must now be traversed at least $p$ times, introducing the extra multiplicative factor. To enforce this, we employ a new semiring consisting of tuples whose coordinates correspond to Hamiltonian paths. 

\begin{restatable}{theorem}{GeneralLowerBound}
\label{thm:mainsemiringcomplexity}
Consider any constant integer $\alpha \geq 1$, and $p \geq 0$. Then, there is a family of instances over a $p$-stable semiring with $\Theta(n)$ nodes and $m=\Theta(n^{1+1/\alpha})$ edges such that, for $p \leq \left(1/\phi\right)^{n^{1/\alpha}/2\alpha}-1$ (where $\phi\in(0,1)$ is some constant) any correct edge schedule must have length
$|\sigma|=\Omega(nm(p+1)).$
\end{restatable}

The missing proofs of both statements can be found in~\apref{\autoref{sec:lowerbound:app}}{the full version of this paper}.

\section{Conclusion}
\label{sec:conclusion}

We believe there is much more to explore about evaluating $\Name$ programs. 
An exact tight iteration bound of $(p+1)n-1$ seems achievable, and is only a factor of $(2+\epsilon)$ away.
Furthermore, beyond fixpoint iteration, bounding the runtime and number of iterations of Newton-like methods is
an open problem~\cite{DBLP:journals/jacm/EsparzaKL10}. 
Lastly, operations beyond semiring operations are useful: 
for example, transitive closure can be solved with fast matrix multiplication. 
As another example, for semirings where the Kleene star
operator is defined ({\em closed} semirings), faster
algorithms like Lehmann's algorithm~\cite{DBLP:journals/tcs/Lehmann77} can be used. 
Finally, beyond $p$-stable semirings, non-stable semirings such as the reals, where convergence can be extended to approximate convergence, are of interest. 
Such semirings are useful in numerical optimization such as gradient descent algorithms.

\bibliographystyle{ACM-Reference-Format}
\bibliography{references}

\appendix
\section{Missing Proofs from Preliminaries}

\begin{proof}[Proof of Observation~\ref{obs:coef-limit}]
Indeed, we have $(p+2) u = (1 \oplus 1 \oplus 1^2 \oplus \dots \oplus 1^{p+1}) u = (1 \oplus 1 \oplus 1^2 \oplus \dots \oplus 1^{p}) u  = (p+1)u$, where the second equality holds because $1$ is a $p$-stable element in $\bm S$.
\end{proof}

\section{Absorption Replacement Process}
\label{appendix:replacementprocess}

This section contains the proofs for Section \ref{sec:absorption} on the absorption
replacement process.

\squeezelemma*
\begin{proof}
    \begin{align*}
        x  & = x \oplus y \oplus z   = x \oplus (p+1) y \oplus (p+1) z &[ x = x \oplus y \oplus z, \text{ applied $p+1$ times}]\\
        & = x \oplus (p+2) y \oplus (p+1) z &  [\text{ Observation } \ref{obs:coef-limit} ]\\
        & = x \oplus (p+1) y \oplus (p+1) z  \oplus y   = x  \oplus y&   \qedhere
    \end{align*}
\end{proof}

\nonzerocorollary*
\begin{proof}
Assume to reach a contradiction that $z \oplus y = 0$. Then applying the Squeeze Lemma with
$x=0$, we can conclude that $y$ must be zero, which is a contradiction.
\end{proof}

\parallelabsorptionlemma*
\begin{proof}
\begin{align*}
  \bigoplus_{i=0}^p \bigoplus_{\bm k : \norm{\bm k}{1}=i} \binom{i}{k_1,\dots,k_r} \bm u^{\bm k}
  &= \bigoplus_{i=0}^p\left(\bigoplus_{j=1}^ru_j\right)^i &[\text{Multinomial Theorem}] \\
  &= \bigoplus_{i=0}^p\left(\bigoplus_{j=1}^ru_j\right)^i + \bigoplus_{i=p+1}^{rp} \left(\bigoplus_{j=1}^ru_j\right)^i & [p\text{-stability of } \bigoplus_{j=1}^ru_j] \\
  &=  \bigoplus_{i=0}^p\left(\bigoplus_{j=1}^ru_j\right)^i  \oplus x & \text{[Squeeze Lemma]}\\
  &= \bigoplus_{i=0}^p \bigoplus_{\bm k : \norm{\bm k}{1}=i} \binom{i}{k_1,\dots,k_r} \bm u^{\bm k} \oplus x &[\text{Multinomial Theorem}]
\end{align*}
Note that the Squeeze Lemma is applicable when it is used because $x$ is a term of $\bigoplus_{i=p+1}^{rp} \left(\bigoplus_{j=1}^ru_j\right)^i$,
since $x$ has $r$ distinct multicands, each $m_i$ is an integer in the range $[0, p]$,
  and $\sum_{i=0}^r m_i \ge p+1$.
\end{proof}

\absorptiontransitivity*
\begin{proof}
\begin{align*}
    \oplus(C) &=  \oplus \left( (A  \multisetminus \{a\})  \uplus B\right)  & \text{[Definition of } C]\\
&=  \oplus \left( ((A  \multisetminus \{a\})  \uplus B) \multisetminus B\right) \oplus \oplus(B) & [a \notin B]\\
&=  \oplus \left( ((A  \multisetminus \{a\})  \uplus B) \multisetminus B\right) \oplus a \oplus \oplus(B)  &\mbox{[$a$ is absorbed by $B$]} \\
&=  \oplus \left( (A  \uplus B) \multisetminus B\right) \oplus  \oplus(B)  &\mbox{[$a \not \in B$]} \\
&=  \oplus ( A  \uplus B) & \mbox{[$((A \uplus B) \multisetminus B) \uplus B = B$]}
\end{align*}

Thus, using the fact that $s$ is absorbed by $A$, and therefore by $A  \uplus B$ as well, we have,
\[
s \oplus  \oplus(C)  =  s \oplus  \oplus ( A  \uplus B)
=   \oplus ( A  \uplus B)
= \oplus(C)
\]
\end{proof}

\section{Background and Intuition for the Linear Case}

The purpose of this section is to give the backstory that
led to our new upper bound analysis of linear $\Name$ convergence. 
The $O((p+1) n^3)$ upper bound analysis from \cite{DatalogoLinear} adopted the following
strategy:
\begin{itemize}
\item Consider a walk $W$ of length at least $(p+2)  n^3$ hops.
    \item Then $W$ is expressed as $W = W' \otimes C^{p+1}$, where $C$ is some subgraph of
    $W$.
    \item Then it is observed that $W_i$, which is equal to $ W$ minus $i$ copies of the
collection $C$ of edges, is a shorter walk than $W$ for all $i \in [1, p+1]$.
    \item Finally, by the $p$-stability of $C$,  it is concluded that $W$ is absorbed by the
    set $A = \{ W_i \mid i \in [1, p+1] \}$ of shorter walks, that is $W \oplus
    \bigoplus_{i=1}^{p+1} W_i = \bigoplus_{i=1}^{p+1} W_i$.
\end{itemize}
The only reasonable candidate for the subgraph $C$ is a closed walk, or more simply a
simple cycle, as $C$ needs to have the property that deleting a copy of $C$ from $W$ results
in a valid walk. To find such a simple cycle $C$, \cite{DatalogoLinear} creates a particular
path-loop decomposition $(P, \mathcal{L})$ of $W$, where $P$ is a simple path, $\mathcal L $ is a
multiset of loops (simple cycles) such that $W = P + \mathcal L$. This path-loop
decomposition has the property that $\mathcal L$ contains $p+2$ copies of some  simple cycle
$C$.

Then \cite{DatalogoLinear} attempted to give a matching $\Omega( (p+1) n^3)$ lower bound
using a collection $G =\{G_{p,n}\}$ and a semiring in which multiplication was intuitively
``maximum''. The lower bound given in \cite{DatalogoLinear} is incorrect.\footnote{This has
been confirmed by the authors of \cite{DatalogoLinear}.} Convergence will occur in
$O((p+1)n)$ steps, for every commutative $p$-stable semiring $\mathbf{S}$ on these graphs.
We will explain (a simplified version of) this $O((p+1)n)$ bound on this key example as it
nicely illustrates the feature of stable semirings that is the basis for our analysis.

\paragraph{Key Example:}
The directed graph $G_{n,p}=(V, E)$ has vertex set $V =[n]$, has  a directed base edge $(i,
 i+1)$ labeled with the multiplicative identity of $\mathbf{S}$, for each $i \in [1, n-1]$,
 and  an edge $(i, j)$ associated with a distinct element $x_{i, j} \ne 1 $ of the domain of $\mathbf{S}$, for $i \in [2n/3, n]$ and $j \in [1, n/3]$. An example with $n=9$
 is shown in Figure \ref{fig:incorrectlowerbound}.

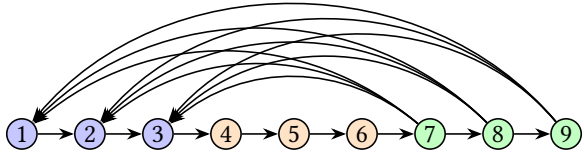
\begin{figure}[H]
    \centering
    \begin{tikzpicture}[node distance=1.5cm,
        every node/.style={circle, draw,minimum size=4mm, inner sep=1pt},
        src/.style={fill=blue!22},
        mid/.style={fill=orange!22},
        dst/.style={fill=green!24},
        ->,
        >=Stealth,  
        semithick,
        scale=0.6,
        ]

        \foreach \i in {1,...,9} {
            \ifnum\i<4
                \node[src] (n\i) at (\i*1.5, 0) {\i};
            \else\ifnum\i<7
                \node[mid] (n\i) at (\i*1.5, 0) {\i};
            \else
                \node[dst] (n\i) at (\i*1.5, 0) {\i};
            \fi\fi
        }

        \foreach \i in {1,...,8} {
            \pgfmathtruncatemacro{\j}{\i + 1}
            \draw (n\i) -- (n\j);
        }

        \foreach \i in {7, 8, 9} {
            \foreach \j in {1, 2, 3} {
                \draw[bend right=5+5*\i] (n\i) to (n\j);
            }
        }

    \end{tikzpicture}

    \caption{Example of the purported lower bound instance from \cite{DatalogoLinear} }
    \label{fig:incorrectlowerbound}
\end{figure}

Let us for simplicity assume that  $p=1$, so for all $x$ in the domain of the semiring
$\bm S$, it is the case that $1 \oplus x \oplus x^2 = 1 \oplus x$, so $x^2$ is absorbed by
$1 \oplus x$. To show that a walk $W$ is absorbed by shorter walks using 1-stability, it
would seem that one would minimally need one multiplier of $W$ to be a square, that is of
the form $x^2$ for some $x$. Yet there are walks of length $\Omega(n^3)$ that contain no
square term, which is presumably why the authors of \cite{DatalogoLinear} believed that the
right bound was cubic in $n$. To see how we can circumvent this issue, let $W$ be a closed
walk starting and ending at vertex $1$ in $G_{n,p}$, consisting of two simple cycles $C_i$
and $C_j$, each containing a single edge with semiring values not equal to 1. Denote these
semiring values as $x_i$ and $x_j$ respectively. Then  $W =  C_i \otimes C_j  =  x_i \otimes
x_j$. But $( C_i \oplus C_j)^2$ is absorbed by $1 \oplus (C_i \oplus C_j) $ by the
$1$-stability of $C_i \oplus C_j$. Thus,  $W =  C_i \otimes C_j $  is absorbed by $1 \oplus
(C_i \oplus C_j) $ as $W =  C_i \otimes C_j$ is merely one term of $( C_i \oplus C_j)^2$. We
can conclude that $W$ is absorbed by shorter walks by noting that for each term of $1 \oplus
(C_i \oplus C_j) $ there is a distinct shorter walk that has product equal to this term. In
short, terms with at least $p+1$ multipliers can be collectively absorbed by exponentially
many ``smaller'' terms. In particular, the product  $\bigotimes_{i=1}^{p+1} x_i$ is absorbed
by the $\frac{(p+1)^{p+1}-1}{p}$ terms in $\bigoplus_{j=0}^p \left( \bigoplus_{i=1}^{p+1} x_i
\right)^j$.

\paragraph{Extension to a General Upper Bound:} 
The key hurdle is how to identify closed walks for a general instance
that serve  the role of the cycles with a single edge not equal to the multiplicative identity in the type of instance
shown in Figure~\ref{fig:incorrectlowerbound}.
The eventual insight was that subtrees in the loop tree 
that are
of the same height,
and that share a common vertex,
serve the roles of the simple cycles in Figure~\ref{fig:incorrectlowerbound}.

\iffullpaper
\section{Convergence Analysis for the Linear Case}

In this section, we complete the analysis of our linear absorption replacement process. We
will show that, starting from a walk $W$  with path-loop decomposition $(P, {\mathcal L})$,
the absorption replacement process terminates with a set  $A$ of walks that are each of
length at most $(2p+1)n$. To accomplish this we need to make some observations about the
absorption replacement process.

\begin{observation}
Let $W'$ be a walk, with associated path-loop decomposition $(P, {\mathcal L}')$, that was an
element of the set $A$ at some point in the absorption replacement process. Then if $L
\in {\mathcal L}'$ then $L \in {\mathcal L}$.
\end{observation}

\begin{proof}
No case in the absorption replacement creates a loop in a path-loop decomposition that was
not present in the prior path-loop decomposition.
\end{proof}

\begin{observation}
\label{obs:maxcopies}
Let $W'$ be a walk, with associated path-loop decomposition $(P, {\mathcal L}')$,
that was an element of the set $A$ at some point in the absorption replacement process.
Then the cardinality of ${\mathcal L}'$ is at most $ p (p+1) $ times the
cardinality of $\mathcal L$.
\end{observation}
\begin{proof}
Every loop in $\mathcal L'$ is contained in $\mathcal L$. If any loop has
multiplicity at least $p+2$ in $\mathcal L$, the first case of the absorption
replacement process will reduce the multiplicity of cycles until all cycles have
multiplicity at most $p+1$. Finally, each iteration of the absorption replacement
process multiplies the number of copies of any loop by at most $p$.
\end{proof}

\begin{observation}
Throughout the absorption replacement process the depth of a loop $L$ is monotonically
nondecreasing.
\end{observation}

\begin{proof}
This is because deleting or replicating loops from $\mathcal L$ can not decrease the depth of any
loop.
\end{proof}

We define a partial ordering $\prec$ on possible loop decompositions created during
the absorption replacement process.

\begin{definition}
Let $\mathcal L' $ and ${\mathcal L}''$
be loop decompositions. We say that ${\mathcal L' } \prec {\mathcal L}'' $
if there exists a $k$ such that for all $i < k$
it is the case that ${\mathcal L }'_i$ and  ${\mathcal L}''_i $
have the same cardinality, and
it is the case that the cardinality of ${\mathcal L }'_k$ is strictly
less than the cardinality of ${\mathcal L}''_k $.
\end{definition}

\begin{lemma}
\label{lem:downward}
Consider a walk $W'$ with associated loop decomposition $\mathcal L'$.
Let $B$ be the set of walks produced from $W'$ in one step of the
absorption replacement process. Let $W''$ be a walk in $B$ with
associated loop decomposition ${\mathcal L}''$. Then ${\mathcal L }'' \prec {\mathcal L}' $.
\end{lemma}
\begin{proof}
In the first case of the absorption replacement process, the lemma follows from the fact that the multiplicity of $L$ strictly decreases
and the multiplicity of every other loop doesn't change.
In the second case and the third case of the absorption replacement process the lemma follows from the fact that the cardinality of
$\mathcal L_d$ decreases, and the cardinality of each $\mathcal L_i$, for $i < d$ does not change.
\end{proof}

\begin{lemma}
This absorption replacement process terminates.
\end{lemma}

\begin{proof}
To aid in our analysis consider the following directed graph $T$.
The nodes of $T$ are the walks in $A$ at some point in the absorption replacement process.
There is a directed edge $(W', W'')$ in $T$ if $W''$ is part of the set $B$ that
is produced from one step of the absorption replacement process on $W'$.
Let $Q$ be a directed walk in $T$ that starts at the original walk $W$ considered
in our absorption replacement process.
Consider the loop cardinality sequence $|\mathcal L_1|, |\mathcal L_2|, \ldots$ of the cardinalities of
the sequence of walks in $Q$. By Lemma \ref{lem:downward}
the partial order on the loop decomposition has to strictly decrease as one traverses $Q$.
Thus $Q$ must be a simple path. Note that by  Observation \ref{obs:maxcopies} for all nodes in
$T$ it is the case that $|\mathcal L_i| \le p(p+1) |\mathcal L|$.
Thus as the maximum  depth in a loop tree is at most $n$, the length of $Q$  is at most
$ (p(p+1) |\mathcal L| )^n$.
\end{proof}

\begin{lemma}
\label{lem:linear-bound}
Let $A$ be the final collection of walks produced by the absorption replacement
process. Let $W'$ be a walk in $A$ with associated loop decomposition $\mathcal L'$.
Then the length of $W'$ is at most $(2p+1) n$.
\end{lemma}

\begin{proof}
Consider an arbitrary vertex $v$ in some vertex level $T_d$. Since the absorption replacement
process terminates, there can be at most $p+1$ loops in $\mathcal L_d$ that intersect $v$,
and at most $p$ loops in $\mathcal L_{d+1}$ that intersect $v$. Thus the number of edges in
$W'$ entering $v$ is at most $2p+1$ and the number of edges in $W'$ exiting $v$ is also at
most $2p+1$. Since an edge exits one vertex and enters one vertex, the total number of edges
in $W'$ can be at most $(2p+1) n$.
\end{proof}

This concludes our proof of Theorem     \ref{thm:linearmain}.
\fi

\iffullpaper
\section{Ancillary Observations for Linear $\Name$}

The following observations about the absorption replacement process for linear $\Name$ are used in the nonlinear analysis.

\begin{observation}
    \label{obs:always-end-with-t}
    When the absorption replacement process replaces a walk $W \in A$ with a collection of walks, $B$, they all end with the same vertex as $W$ does.
\end{observation}

\begin{observation}
    \label{obs:lowdegree}
    Let $A$ be the final collection of walks produced by the absorption replacement
    process. Let $W'$ be a walk in $A$ and $v$ a vertex in $W'$.
    Then the in-degree and out-degree of $v$ in $W'$ are at most $2p+1$. Thus, $v$ occurs in $W'$ at most $2p+1$ times.
\end{observation}

\begin{proof}
Then if $v$ had  in-degree or
    out-degree greater than $2p+1$ then we would be able to apply either case 2
    or case 3 in the construction of an absorbing set $B$ to $v$.
\end{proof}

\begin{observation}
    \label{obs:NoMoveAcrossV}
Let $W'$ be a walk that was in the set $A$ at some point in the absorption
replacement process, and let $B$ be the collection derived from $W'$ in one step of the absorption replacement process.
\begin{itemize}
    \item If  a vertex $v$ in $W'$ is visited only once then
for every walk $W'' \in B$ it is the case that $v$ is only visited once.
\item
If  a vertex $v$ in $W'$ is visited only once, and the vertices visited
before $v$ in $W'$ are disjoint from the vertices visited after $v$ in $W'$,
then
for every walk $W'' \in B$ it is the case that vertex $v$ in $W''$ is visited only once, and
the vertices visited
before $v$ in $W''$ are disjoint from the vertices visited after $v$ in $W''$.
\end{itemize}
\end{observation}

\begin{proof}
  Both statements immediately follow from the observation that  $W''$ is created from $W'$ by duplicating some loops, and deleting other loops.
\end{proof}
\fi

\iffullpaper
\section{Upper Bounds for General Datalogo Programs}
\label{sec:generalupperbound:app}

In this section we give the full proof of Theorem~\ref{thm:nonlinearmain}. We prove the
theorem by defining an absorption replacement process that produces a collection $A$ of
shallow (derivation) trees that collectively absorbs  an arbitrary deep (derivation) tree
$T$, essentially by reduction to the linear case. This process starts by transforming $T$
into what we call a canonical tree, which we explain in subsection \ref{sec:transform:app}. To
aid  in explaining  this transformation, we first define what we call the path tree $L(T)$
of $T$ in subsection \ref{subsec:leftpathtree:app}. Then in subsection
\ref{subsec:reductiontolinear:app} we explain what it means to apply the linear absorption
replacement process to a derivation tree. Then in subsection
\ref{subsec:nonlinearreplacementprocess:app} we explain how to construct the multiset $B$ of
(derivation) trees from a single (derivation) tree $T$ by applying the linear
absorption replacement process as a black box. Finally, we show that this process terminates
in a collection $A$ of shallow trees that collectively absorb $T$ in subsection
\ref{sec:nonlinearanalysis:app}.

\subsection{ Path Tree}
\label{subsec:leftpathtree:app}

Consider a derivation tree $T$.  We will now define a new tree $L(T)$, which we call a \textit{path tree}, that is derived from $T$. This path tree allows us to view a derivation tree as a
hierarchical collection of paths/walks. This is crucial, as it will allow us to use the linear absorption from the previous section as a black box. Each node in $L(T)$ corresponds to some
ancestor-to-descendant path in $T$, but not all ancestor-to-descendant paths in $T$ will be
in $L(T)$.

\pathtreedef*

\begin{figure}[h!]
\centering
\definecolor{ptRed}{HTML}{E27D7D}
\definecolor{ptGreen}{HTML}{7FB069}
\definecolor{ptBlue}{HTML}{6FA8DC}
\definecolor{ptYellow}{HTML}{F5C26B}
\begin{minipage}[b]{0.40\textwidth}
\centering
\scalebox{0.65}{\begin{tikzpicture}
[node distance=2.5mm and 1.5mm,
every node/.style={circle, draw,minimum size=6.5mm, inner sep=1pt},semithick]
\node[fill=ptRed] (n0) {1};
\node[fill=ptRed, below left=of n0] (n1) {2};
\node[fill=ptGreen, below right=of n0] (n2) {3};
\node[fill=ptRed, below left=of n1] (n3) {4};
\node[fill=ptGreen, below left=of n2] (n5) {6};
\node[fill=ptYellow, below right=of n2] (n6) {7};
\node[fill=ptBlue, below right=of n3] (n4) {5};
\node[fill=ptRed, below left=of n3] (n7) {8};
\draw (n0) -- (n1);
\draw (n0) -- (n2);
\draw (n1) -- (n3);
\draw (n3) -- (n4);
\draw (n2) -- (n5);
\draw (n2) -- (n6);
\draw (n3) -- (n7);
\end{tikzpicture}}
\end{minipage}
\hspace*{0.1\textwidth}
\begin{minipage}[b]{0.40\textwidth}
\centering
\scalebox{0.65}{\begin{tikzpicture}
[node distance=2.5mm and 1.5mm,
every node/.style={circle, draw,minimum size=6.5mm, inner sep=1pt},semithick]
\node[fill=ptRed, circle] (r) {1,2,4,8};
\node[fill=ptGreen, circle, below left=of r] (g) {3,6};
\node[fill=ptBlue, circle, below right=of r] (b) {5};
\node[fill=ptYellow, circle, below left=of g] (y) {7};
\draw (r) -- (g);
\draw (r) -- (b);
\draw (g) -- (y);
\end{tikzpicture}}
\end{minipage}
\caption{The left figure shows a derivation tree $T$ and the right shows its path tree $L(T)$. This example shows $4$ ancestor-to-descendant paths of $T$ and their $4$ corresponding path-nodes in $L(T)$.
A root-to-leaf path of $T$ can be thought of as a path from the root to some path-node in the decomposition $L(T)$, where only a prefix of the path of every path-node is traversed.
Here, $P(T) = (1,2,4,8)$.
}
\label{fig:tree-decomposition:app}
\end{figure}
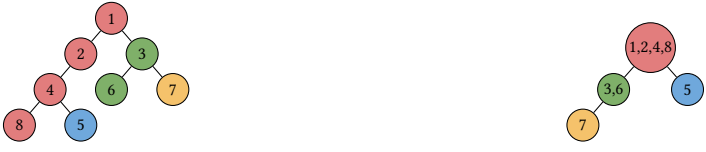

Each node in $L(T)$ corresponds to a node-to-leaf path in $T$. To avoid confusion between
derivation trees and path trees, we will call a node in a derivation tree $T$ a
\textit{node} and a node in a decomposition $L(T)$ a \textit{path-node}.
Note that a path-node corresponds to multiple nodes in the derivation
tree, which form the left-most path of some sub-tree of the derivation tree.
When we write path-node, the reader should think of this corresponding path in $T$.
Figure~\ref{fig:tree-decomposition:app} illustrates this decomposition.

\subsection{Canonical Tree Transformation}
\label{sec:transform:app}

Here we will describe a procedure to transform any derivation tree into what we call a canonical tree. This transformation reorders different subtrees such that the tree becomes as ``left-deep'' as possible. This is useful, because the linear absorption procedure that we will use as a black-box works on paths/walks, so it is more effective on left-deep trees.

Before we state the canonical tree construction process in Algorithm~\ref{alg:canonicaltransform:app}, we need some definitions related to wedges. See Figure \ref{fig:wedgepush:app} for an illustration.

\wedgedef*

\movablewedgedef*

We can move a movable wedge to a higher left-most path using the following wedge push
operation.

\wedgepushdef*

Figure~\ref{fig:wedgepush:app} illustrates a wedge push.
Note that the operation is still well defined when $v = u$. Then, either $\alpha$ or $\beta$ is empty, and we do not need to change the permutation of the spine. We can now define our canonical tree transformation procedure.

\begin{figure}[t] 
\def\xoff{-0.6}
\def\yoff{-0.65}
\begin{minipage}[c]{0.42\textwidth}
\centering
\scalebox{0.55}{
\begin{tikzpicture}[node distance=0.5cm,
every node/.style={minimum size=0mm, inner sep=0pt},
->,>=Triangle,thick]
\tikzset{
trianglenode/.style={
regular polygon,
regular polygon sides=3,
minimum size=0.7cm,
inner sep=0pt,
text centered
},
wedge/.style={draw=cyan!70!black,fill=cyan!18},
spine/.style={draw=cyan!70!black}
}

\node (n0) {};
\node (n1) at ($(n0)+(\xoff,\yoff)$) {$v$};
\node(n2) at ($(n1)+(\xoff,\yoff)$) {};
\node[trianglenode,draw] (n3) at ($(n1)+(-\xoff,\yoff)$) {$A$};

\draw(n0) -> (n1);
\draw (n1) -> (n2);
\draw (n1) -> (n3.north);

\node (m0) at ($(n0) + (0.7,-1.4)$) {$u$};
\node[trianglenode,wedge] (tB) at ($(m0)+(-\xoff,\yoff)$) {$B$};
\node (m1) at ($(m0)+(\xoff,\yoff)$) {$v$};
\node[trianglenode,wedge] (tC) at ($(m1)+(-\xoff,\yoff)$) {$C$};
\node (m2) at ($(m1)+(\xoff,\yoff)$) {$u$};
\node[trianglenode,draw] (tD) at ($(m2)+(-\xoff,\yoff)$) {$D$};
\node (wedge) at ($(tC)+(0.75,0)$) {{\color{cyan!70!black} $W$}};

\draw[spine] (m0) -> node[pos=0.4,left, xshift=-2pt,text=cyan!70!black] {$\alpha$} (m1);
\draw[spine] (m1) -> node[pos=0.4,left, xshift=-2pt,text=cyan!70!black] {$\beta$} (m2);
\draw[spine] (m0) -> (tB.north);
\draw[spine] (m1) -> (tC.north);
\draw (m2) -> (tD.north);
\draw[dashed,-] (n0) to[out=-25, in=160] (m0);

\node(A) at ($(n0)+(0,0.6)$)  {};
\node(B) at ($(A)+(-4.06,-4.4)$) {};
\node (C) at ($(A)+(4.06,-4.4)$) {};
\draw[thick,-] (A) -- (B) -- (C) -- (A);
\end{tikzpicture}
}
\end{minipage}%
\hfill
\begin{minipage}[c]{0.06\textwidth}
\centering
$\Longrightarrow$
\end{minipage}%
\hfill
\begin{minipage}[c]{0.42\textwidth}
\centering
\scalebox{0.55}{
\begin{tikzpicture}[node distance=0.5cm,
every node/.style={minimum size=0mm, inner sep=0pt},
->,>=Triangle,thick]
\tikzset{
trianglenode/.style={
regular polygon,
regular polygon sides=3,
minimum size=0.7cm,
inner sep=0pt,
text centered
},
wedge/.style={draw=cyan!70!black,fill=cyan!18},
spine/.style={draw=cyan!70!black}
}

\node (n0) {};
\node (n1) at ($(n0)+(\xoff,\yoff)$) {$v$};
\node[trianglenode,wedge] (tC) at ($(n1)+(-\xoff,\yoff)$) {$C$};
\node (n2) at ($(n1)+(\xoff,\yoff)$) {$u$};
\node[trianglenode,wedge] (tB) at ($(n2)+(-\xoff,\yoff)$) {$B$};
\node (n3) at ($(n2)+(\xoff,\yoff)$) {$v$};
\node[trianglenode,draw] (tA) at ($(n3)+(-\xoff,\yoff)$) {$A$};
\node(n4) at ($(n3)+(\xoff,\yoff)$) {};

\draw(n0) -> (n1);
\draw[spine] (n1) -> node[pos=0.4,left, xshift=-2pt,text=cyan!70!black] {$\beta$} (n2);
\draw[spine] (n2) -> node[pos=0.4,left, xshift=-2pt,text=cyan!70!black] {$\alpha$} (n3);
\draw (n3) -> (n4);

\node (m0) at ($(n0) + (0.7,-1.4)$) {$u$};
\node[trianglenode,draw] (tD) at ($(m0)+(-\xoff,\yoff)$) {$D$};

\draw[spine] (n1) -> (tC.north);
\draw[spine] (n2) -> (tB.north);
\draw (n3) -> (tA.north);
\draw[dashed,-] (n0) to[out=-25, in=160] (m0);
\draw (m0) -> (tD.north);

\node(A) at ($(n0)+(0,0.6)$)  {};
\node(B) at ($(A)+(-4.06,-4.4)$) {};
\node (C) at ($(A)+(4.06,-4.4)$) {};
\draw[thick,-] (A) -- (B) -- (C) -- (A);

\end{tikzpicture}
}
\end{minipage}
\caption{Illustration of a wedge push, where a wedge whose spine consists of subpaths $\alpha$ and $\beta$ is pushed to the left black path (or path-node) $P$. The triangles represent subtrees hanging off the path or the spine. Note that the spine is a descendant of $P$ and may intersect multiple path-nodes.}
\label{fig:wedgepush:app}
\end{figure}
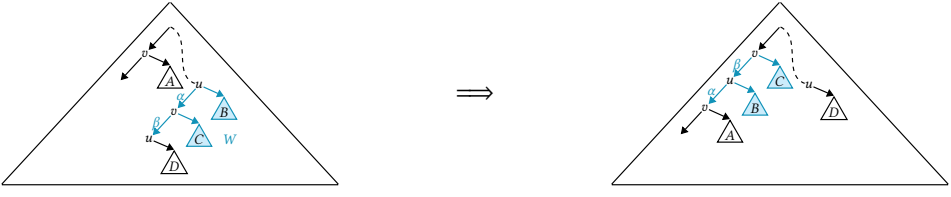

\begin{algorithm}[H]
    \caption{Canonical Tree Transformation Procedure}
    \label{alg:canonicaltransform:app}
    \small
    \begin{algorithmic}[1]
        \REQUIRE Input derivation tree $T$
        \STATE \textbf{Procedure} Transform($T$)
        \WHILE{$\exists P \in L(T)$ with a movable wedge $W$ with respect to a node $v$}
            \STATE Perform a wedge-push on $W$ to $P$ with respect to $v$
        \ENDWHILE
        \RETURN $T$
        \STATE \textbf{End Procedure}
    \end{algorithmic}
\end{algorithm}

\cantreetransformterminate*
\begin{proof}
Let $T$ be the input derivation tree.  Let $\ell_i$ be the total number of nodes in $T$
which are in some path-node at level $i$ in $L(T)$. Consider the lexicographical order of
this vector, which is a total order. This order increases strictly with each execution of
the while loop of the procedure, since some wedge is moved higher in $T$.
\end{proof}

\begin{restatable}{proposition}{gencanonicalsameyield}
Let $T$ be the tree before applying the canonical tree transformation procedure and $T'$ be the tree after. The yields of $T$ and $T'$ are the same.
\end{restatable}

The following lemma states some structural properties of a canonical tree. The intuition here is that if these properties do not hold, we can do another wedge push.

\PostTransProcedureProps*
\begin{proof}
\textbf{(1)} Suppose for sake of contradiction there is a path-node $P_i$ with $i \geq 2$ where the
corresponding path in $T$ has two copies of $v$. By definition, there is then a movable
wedge with respect to $v$ and $P_1$. This is the wedge rooted at $v$ in $P_i$ and ending at
the other copy of $v$ in $P_i$. Due to this, this cannot be a tree output by the canonical
tree transformation procedure.

\textbf{(2)} Fix a path-node $P_i$ where $i \geq 2$. Suppose for sake of contradiction there
is a node $v'$ occurring both before and after $v$ in the path corresponding to $P_i$. Then,
there is a movable wedge $W$ from $v'$ to $v'$ that includes $v$ with respect to $P_1$. The
canonical tree transformation procedure must have extracted $W$ and moved it to a copy of
$v$ in the path corresponding to the $P_1$ in $T$, a contradiction.
\end{proof}

We finally observe that we can keep having compatibility, which will be necessary to apply the linear absorption replacement process.

\begin{observation}
    If we apply Wedge Push on a compatible tree, we obtain a compatible tree with the same product of the edge labels. Therefore, the Canonical Tree Transformation outputs a compatible tree with the same yield. 
\end{observation}

\subsection{Application of the Linear Absorption Replacement Process to Derivation Trees}
\label{subsec:reductiontolinear:app}

Given our assumptions about the structure of derivation trees, our linear
absorption replacement process $\textsf{LinAbsorb}(P)$ can be applied to any
ancestor-to-descendant path $P$ of a derivation tree $T$. This creates a collection $A$ of
paths, derived from $P$, which absorb $P$. We will explain how to use
$A$ and $T$ to create a collection of trees $A'$ that absorb $T$.
Let $(X, Y)$ be a directed edge in $P$, where $X$ is the parent of $Y$ in $T$.
Let $T_{X, 1}, \ldots, T_{X, k}$ be the subtrees of the children of $X$ in $T$, other than
$Y$. We will think of the subtrees  $T_{X, 1}, \ldots, T_{X, k}$ as being
associated with the edge $(X, Y)$ in $T$; when \textsf{LinAbsorb} duplicates $(X, Y)$, the
subtrees  $T_{X, 1}, \ldots, T_{X, k}$ are also duplicated, with each copy of $X$
getting duplicate copies of the subtrees as children. Similarly when \textsf{LinAbsorb} deletes $(X, Y)$ the
subtrees $T_{X, 1}, \ldots, T_{X, k}$  are deleted.
Note for completeness that the trees $A'$ produced are all compatible.
See Figure~\ref{fig:path-subtrees:app} for
an illustration, where for simplicity we assume $k$ is universally 1.

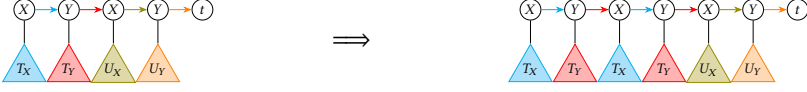
\begin{figure}
\centering
\begin{minipage}[c]{0.41\textwidth}
\centering
\scalebox{0.55}{
\begin{tikzpicture}
[node distance=0.55cm,
        every node/.style={circle, draw,minimum size=5mm, inner sep=1pt},
        ->,
        >=Stealth,  
        semithick,
        scale=0.5,
        box/.style={draw, rectangle, rounded corners, thick, red, inner sep=5pt}
        ]
     \tikzset{
    trianglenode/.style={
      regular polygon,
      regular polygon sides=3,
      minimum size=1.2cm,
      inner sep=0pt,
      text centered
    }}
\node (a) {$X$};
\node[right=of a] (b) {$Y$};
\node[right=of b] (c) {$X$};
\node[right=of c] (d) {$Y$};
\node[right=of d] (e) {$t$};

\node[trianglenode,draw=cyan, fill=cyan!30,below=of a] (ta) {$T_X$};
\node[trianglenode,draw=red, fill=red!30,below=of b] (tb) {$T_Y$};
\node[trianglenode,draw=olive, fill=olive!30,below=of c] (tc) {$U_X$};
\node[trianglenode,draw=orange, fill=orange!30,below=of d] (td) {$U_Y$};

\draw[cyan,->] (a) -- (b);
\draw[red,->] (b) -- (c);
\draw[olive,->] (c) -- (d);
\draw[orange,->] (d) -- (e);

\draw[-] (a) -- (ta);
\draw[-] (b) -- (tb);
\draw[-] (c) -- (tc);
\draw[-] (d) -- (td);

\end{tikzpicture}
}
\end{minipage}%
\hfill
\begin{minipage}[c]{0.05\textwidth}
\centering
$\Longrightarrow$
\end{minipage}%
\hfill
\begin{minipage}[c]{0.53\textwidth}
\centering
\scalebox{0.55}{
\begin{tikzpicture}
[node distance=0.55cm,
        every node/.style={circle, draw,minimum size=5mm, inner sep=1pt},
        ->,
        >=Stealth,  
        semithick,
        scale=0.5,
        box/.style={draw, rectangle, rounded corners, thick, red, inner sep=5pt}
        ]
     \tikzset{
    trianglenode/.style={
      regular polygon,
      regular polygon sides=3,
      minimum size=1.2cm,
      inner sep=0pt,
      text centered
    }}
\node (a) {$X$};
\node[right=of a] (b) {$Y$};
\node[right=of b] (c) {$X$};
\node[right=of c] (d) {$Y$};
\node[right=of d] (e) {$X$};
\node[right=of e] (f) {$Y$};
\node[right=of f] (g) {$t$};

\node[trianglenode,draw=cyan, fill=cyan!30,below=of a] (ta) {$T_X$};
\node[trianglenode,draw=red, fill=red!30,below=of b] (tb) {$T_Y$};
\node[trianglenode,draw=cyan, fill=cyan!30,below=of c] (tc) {$T_X$};
\node[trianglenode,draw=red, fill=red!30,below=of d] (td) {$T_Y$};
\node[trianglenode,draw=olive, fill=olive!30,below=of e] (te) {$U_X$};
\node[trianglenode,draw=orange, fill=orange!30,below=of f] (tf) {$U_Y$};

\draw[cyan,->] (a) -- (b);
\draw[red,->] (b) -- (c);
\draw[cyan,->] (c) -- (d);
\draw[red,->] (d) -- (e);
\draw[olive,->] (e) -- (f);
\draw[orange,->] (f) -- (g);

\draw[-] (a) -- (ta);
\draw[-] (b) -- (tb);
\draw[-] (c) -- (tc);
\draw[-] (d) -- (td);
\draw[-] (e) -- (te);
\draw[-] (f) -- (tf);

\end{tikzpicture}
}
\end{minipage}
\caption{Illustration of duplication (and possible removal) of subtrees when \textsf{LinAbsorb} is applied to a path (or equivalently, a path-node). The left shows a node-to-leaf path $P$ in $T$ and subtrees, shown as triangles. 
Note that $P$ must end with the special sink $t$ as we consider compatible trees wlog. \textsf{LinAbsorb} takes as input a walk ending with $t$. 
While $T_X$ is a subtree of $X$, we treat it as being associated with edge $(X,Y)$. Note that two copies of the same edge can have different subtrees. The right shows a new path/walk created by \textsf{LinAbsorb} along with duplicated subtrees. When the first $(X,Y)$ is duplicated, its subtree $T_X$ is also duplicated.
}
\label{fig:path-subtrees:app}
\end{figure}

\subsection{Absorption Replacement Process }
\label{subsec:nonlinearreplacementprocess:app}

We here describe one step of our absorption replacement process (Algorithm~\ref{alg:absorb}), which produces a collection $B$ of trees from a tree $T$. For ease of reading, we restate the algorithm.
Each step applies the linear absorption replacement process to a path-node in $L(T)$.
The path-node selected has the lowest depth among all path-nodes whose corresponding path has length greater than $n(2p+1)$.
In Algorithm~\ref{alg:absorb}, $\text{length}(P)$ is the number of edges in $P$.

\begin{algorithm}[H]
    \caption*{{\bf Algorithm \ref{alg:absorb}} Absorption Replacement}
    \small
    \begin{algorithmic}[1]
        \REQUIRE Input derivation tree $T$
        \STATE \textbf{Procedure} Absorb($T$)
            \STATE $T \gets$ Canonical Transformation of $T$
            \IF {$\exists P \in L(T): \text{length}(P) > n(2p+1)$}
                \STATE $P \gets$ such a path  $P \in L(T)$ of lowest depth
                \STATE $P_1,\dots,P_k \gets \text{\textsf{LinAbsorb}}(P)$ %
                \STATE $T_1,\dots,T_k \gets \text{Produce trees from } T \text{ with } P_1,\dots,P_k$ %
                \STATE $B \gets  \{ T_1,\dots,T_k \}$
            \ENDIF
            \STATE \textbf{return} $B$ %
    \end{algorithmic}
\end{algorithm}

\subsection{Analysis of Replacement Process}
\label{sec:nonlinearanalysis:app}

To establish Theorem~\ref{thm:nonlinearmain}, we need to prove that the absorption replacement process terminates, and that the trees in the resulting set are not too deep.

\GeneralAbsReplacementTermination*
\begin{proof}
The process recursively applies \textsf{LinAbsorb} on path-nodes in order of increasing depth. Each call only modifies subtrees hanging off the input path (without increasing the depth of $L(T)$) and outputs a path of length $\leq (2p+1)n$.
\end{proof}

\begin{observation}
\label{ob:appeared-before}
Suppose that we have a subpath $\beta$  consisting of path-nodes in $L(T)$, and $T$ was
obtained from $T'$ by one step of the absorption replacement process when \textsf{LinAbsorb} was run
on a path-node not in $\beta$. Then, $L(T')$ must have the same subpath $\beta$.
\end{observation}
\begin{proof}
When \textsf{LinAbsorb} is run on a path-node $P'$, some subtrees hanging off the path may be removed
or replicated, but no new subtrees are created. Further, it has no effect on the ascendant
path-nodes of $P'$.
\end{proof}

We will now show that after the replacement process terminates, all the resulting trees,
whose set is denoted $A$, have a small depth. Specifically, we will show that if the linear
absorption black box generates walks with length at most $(2p+1)n$ (each node is visited at most
$2p+1$ times), then the depth of any tree in $A$ is at most $(2p+2)n$. This
will establish a convergence time of $2(p+1) n$ for a naive evaluation of an arbitrary
$\Name$ instance.

\KeyPropertiesAfterAbsorptionGeneral*
\begin{proof}
We first prove Property (2) by arguing that any path-node can have at most $2p+1$ occurrences
of $v$. Note that \textsf{LinAbsorb} is run on a path-node of the smallest depth in the path tree.
Thus, the property follows since when \textsf{LinAbsorb} is run on a path-node $P'$, it only affects
itself and its descendant path-nodes, and the new path $P$ generated from $P'$ has at most
$2p+1$ occurrences of any node thanks to Observation~\ref{obs:lowdegree}.

To prove Property (1) we need to take a close look at how the collection $A$ of the trees
evolves over iterations. Suppose $T$ was obtained by the following sequence of the
absorption replacement process: $T_0$, $T_1, \ldots, T_k = T$, which means that $T_i$ was
replaced by some trees including $T_{i+1}$ for all $i \in \{0, 1, \ldots, k-1\}$. We will
prove Property (1) by induction on $i$.

For the base case, consider $T_0$, which is the immediate outcome of Transform performed on
the input tree. Property (1) is an immediate consequence of
Lemma~\ref{lemma:PostTransProcedureProps}.

Assume that $T' = T_{k-1}$ satisfies Property (1), but $T$ does not. Suppose that for some
root-leaf path $Q$ in $T$, there are path-nodes $P_1$ and $P_2$ in $\beta(Q)$, from top to
bottom, such that $P_1$ includes $v$ and $s(P_2)$ includes $v$;  otherwise the property
holds true vacuously. Suppose $T$ was obtained from $T'$ by performing \textsf{LinAbsorb} on a
path-node $P'$ in $T'$ and let $P$ be the outcome when \textsf{LinAbsorb} was run on $P'$; we will
call $P'$ the precursor of $P$.

If $P \neq P_1, P_2$, by Observation~\ref{ob:appeared-before}, $L(T')$ must have a subpath
consisting of $P_1$ and $P_2$ (from top to bottom). By the induction hypothesis, i) and ii)
continue to hold true.
For the same reason if $P \neq P_1, P_2$ and $P$ is an ascendant of $P_2$, iii) holds true.
If $P$ is a descendant of $P_2$, consider the subtree $\tau$ hanging off $P_2$ on $u$ in
$T$; $\tau$ is disjoint from $P_2$. Similarly, let $\tau'$ be the subtree hanging off $P_2$
on $u$ in $T'$. The induction hypothesis with $P_1$ and $P_2$ for $T'$ implies that $\tau'$
doesn't contain $v$. Since \textsf{LinAbsorb} can replicate a node but never creates a node that
didn't exist, $\tau$ doesn't include $v$.

Thus, it now remains to show that $T$ satisfies Property iii) of (1) when $P$ is $P_1$ or
$P_2$. If $P = P_1$, the precursor $P'_1$ of $P_1$ must have $v$. Thus, using the induction
hypothesis with $P'_1$ and $P_2$ for $T'$, we have Property (1) for $T$.

We now consider the most interesting case $P  = P_2$. By the induction hypothesis, we know
that $P_2'$, the precursor of $P_2$, has exactly one occurrence of $v$ and the nodes above
$v$  are disjoint from those below $v$ on $P'_2$ . Then, by
Observation~\ref{obs:NoMoveAcrossV}, i) and ii) immediately follow. For the sake of
contradiction, suppose iii) doesn't hold for $T$. Then, it means that there exists a subtree
$\tau$ hanging off $P_2$ on $u$, which appears below $v$ on $P_2$, that contains $v$.
However, by Observation~\ref{obs:NoMoveAcrossV} and the induction hypothesis regarding $P_1$
and $P'_2$ for $T'$, $u$ can only occur below $v$ on $P'_2$. Since $\tau$ is associated with
an occurrence of $v$, $\tau$ must hang off $P'_2$ on $u$ that appears below $v$. This
implies that $T'$ violates iii) of  Property (1) for $P_1$, $P_2$, and the path-node $P_3$ in
$L(\tau)$ ($P_3$ is also a path-node in $L(T)$), which is a contradiction.
\end{proof}

The key idea behind (in particular iii) of) Property $(1)$ is
Observation~\ref{obs:NoMoveAcrossV}. Since the nodes before and after the single occurrence
of $v$ on $P_2$ are disjoint, the absorption replacement process cannot move wedges across
$v$. We illustrate the idea with the following example.

\paragraph{Example}
This example illustrates Property (1) of the previous lemma. Consider the following derivation tree.
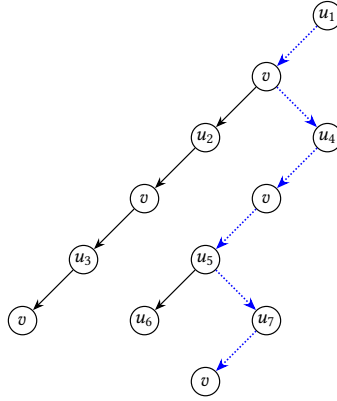
\begin{center}
\scalebox{0.75}{\begin{tikzpicture}[node distance = 1cm,
    every node/.style={circle, draw,minimum size = 5mm,inner sep=1pt},
    ->,
    >=Stealth,  
    semithick,
    scale=0.8,
    ]
    \node[circle] (n0) {$u_1$};
    \node[circle, below left=of n0] (n1) {$v$};
    \node[circle, below left=of n1] (n2) {$u_2$};
    \node[circle, below left=of n2] (n3) {$v$};
    \node[circle, below left=of n3] (n4) {$u_3$};
    \node[circle, below left=of n4] (n5) {$v$};

    \node[circle, below right=of n1] (n6) {$u_4$};
    \node[circle, below left=of n6] (n7) {$v$};
    \node[circle, below left=of n7] (n8) {$u_5$};
    \node[circle, below left=of n8] (n9) {$u_6$};

    \node[circle, below right=of n8] (n10) {$u_7$};
    \node[circle, below left=of n10] (n11) {$v$};

    \draw[blue,thick,densely dotted] (n0) -- (n1);
    \draw (n1) -- (n2);
    \draw (n2) -- (n3);
    \draw (n3) -- (n4);
    \draw (n4) -- (n5);
    \draw[blue,thick,densely dotted] (n1) -- (n6);
    \draw[blue,thick,densely dotted] (n6) -- (n7);
    \draw[blue,thick,densely dotted] (n7) -- (n8);
    \draw (n8) -- (n9);
    \draw[blue,thick,densely dotted] (n8) -- (n10);
    \draw[blue,thick,densely dotted] (n10) -- (n11);
\end{tikzpicture}}
\end{center}
Property $(1)$ argues that something like depicted above cannot happen, where there is a
wedge from $v$ to $v$ that appears between the second and third left-most paths, where $v$
also appears in the first left-most path. In other words, a root-to-leaf path $Q$ in $T$
colored blue corresponds to three path-nodes, $P_1 = (u_1, v, u_2, v, u_3, v)$, $P_2 = (u_4,
v, u_5, u_6)$, and $P_3 = (u_7, v)$. Note that $P_1$ has multiple occurrences of $v$ and
$P_2$ has exactly one occurrence of $v$. So, $T$ satisfies i) and ii) of Property (1), but
violates iii) of Property (1) because of $P_3$ including $v$.

This cannot happen for the following reason. Since $v$ occurs exactly once on $P_2$ when the
linear absorption is run on $P_2$, $u_5$ can't move across $v$ on $P_2$
(Observation~\ref{obs:NoMoveAcrossV}). By recursively applying this argument, we can
conclude that $u_5$ must appear below $v$ on (the path corresponding to $P_2$) in the
initial canonical tree. However, this is not possible since the canonical transformation
would have moved the wedge defined by the wedge-path $(v, u_5, u_7, v)$ to $P_1$.

\begin{lemma}
\label{lem:nonlinear-final-in-t:app}
When the absorption replacement process terminates, any resulting derivation tree has depth at most $n(2p+2)$.
\end{lemma}
\begin{proof}
Consider any such tree $T$ and any root-to-leaf path $Q$ in $T$. Let $\beta$ be the path in
$L(T)$ corresponding to $Q$. Because of iii) of Property (1) of
Lemma~\ref{lem:non-linear-key}, there are at most two path-nodes in $\beta$ including $v$.
Further, the first and second path-nodes (if they exist) can have at most $(2p+1)$ and $1$
occurrences of $v$, respectively. Thus, $Q$ can have at most $2p+2$ occurrences of $v$.
\end{proof}

From Lemma~\ref{lem:linear-bound}, Observation~\ref{obs:lowdegree}, and Lemma~\ref{lem:nonlinear-final-in-t:app}, we have established the upper bound for the nonlinear case claimed in Theorem~\ref{thm:nonlinearmain}.
\fi

\iffullpaper
\section{Lower Bounds}
\label{sec:lowerbound:app}

Section~\ref{sec:prelimlem}
gives helper lemmas that will be used in the analysis.  Section~\ref{sec:sparse} gives a
recursive construction of sparse graphs and shows a lower bound for them on a $0$-stable
semiring,  establishing a lower bound for arbitrary $m$ on $0$-stable semirings.  Finally,
Section~\ref{sec:general} builds on this to establish our most general lower bound.

\subsection{Why Relaxation Algorithms?}

The motivation behind studying the class of relaxation algorithms is that most known algorithms that evaluate $\Name$ follow this restriction. An important detail is that in our definition of relaxation algorithms, every relaxation step updates only one variable with the most recent value of another variable, while keeping all other values the same. However, algorithms such as na\"ive and semi-na\"ive perform a ``batch'' update and update simultaneously multiple variables. We note here that because all $p$-stable semirings are naturally ordered, it never hurts in terms of convergence to use the most recent value of a variable to perform an update.

{ Na\"ive evaluation}~\cite{DBLP:books/aw/AbiteboulHV95} is a relaxation algorithm, where at every iteration we relax all edges in the graph (in any order).
Semi-na\"ive evaluation~\cite{DBLP:books/aw/AbiteboulHV95} is also a relaxation algorithm, where the edge-schedule has iterations
similar to naïve evaluation. However, each iteration does not relax all edges in the graph, but only the
subset where the variable value was updated in the previous iteration. This
creates a possibly shorter edge-schedule than the corresponding naïve schedule.

Both na\"ive and semi-na\"ive evaluation work for any $p$-stable semiring. However, there are specialized algorithms that take advantage of specific semiring properties to converge faster
\cite{GroundingBasedApproach}. For $\Name$ over totally ordered absorptive semirings,
evaluation can be done with $O(m\log m)$ semiring operations (such as the tropical semiring), using a length $O(m)$ edge
schedule, with an algorithm similar to Dijkstra's algorithm~\cite{CLRS}. Similarly, for semirings of
rank $r$, evaluation can be done with $O(rm)$ semiring operations, using an edge schedule of
the same length. These algorithms are both relaxation algorithms.


Finally, we note that in our definition of relaxation algorithms, we only consider linear $\Name$. It is also possible (but outside the scope of this paper) to generalize relaxation algorithms to work on general $\Name$ programs, by relaxing a monomial-at-a-time instead of an edge-at-a-time.


\subsection{Helpful Preliminary Lemmas}
\label{sec:prelimlem}

We present here two combinatorial lemmas that are useful for the remaining lower bound proofs.

The absorptive lower bounds essentially bounds the length of a schedule that covers a set of paths one time. This lemma bounds the length to cover the same set of paths, but multiple times.
If $\sigma$ is a sequence of edges and $q\geq 1$ an integer, we let $\sigma^{(q)}$ denote the sequence of edges obtained by repeating $\sigma$, $q$ times. If $H$ is a set of sequences, then we also define $H^{(q)} := \{\sigma^{(q)}\mid \sigma \in H\}$.

 \begin{lemma} \label{lemma:InductiveScheduleAppend}
 Let $H$ be a set of edge sequences. Suppose that for every $\theta\in(0,1]$, any edge schedule that covers at least a fraction $\theta$ of the paths in $H$ has length $> L(\theta)$, where $L:(0,1]\rightarrow\mathbb{Z}^{>0}$ is an increasing function. Let $q>1$ be an integer and $\theta_q\in(0,1]$. Then, any edge schedule that covers at least a fraction $\theta_q$ of $H^{(q)}$ has length at least $q \cdot L(\theta_q/q)$.
 \end{lemma}

\begin{proof}
For the sake of contradiction, suppose that there is an edge schedule $\sigma$ with length $|\sigma| \leq q \cdot L(\theta_q/q)$ that covers at least a $\theta_q$-fraction of $H^{(q)}$. We can split $\sigma$ into $k$ consecutive subsequences $\sigma_1, \dots, \sigma_k$, such that $k \leq q$ and for each $i, |\sigma_i| \leq L(\theta_q/q)$.

From our assumption in the lemma, every subsequence $\sigma_i$ covers strictly less than a $(\theta_q)/q$-fraction of $H$. Let $\bar{H}$ be the set of sequences from $H$ covered by some subsequence $\sigma_i$: we then have $|\bar{H}| < k\cdot (\theta_q/q) \cdot |H| \leq \theta_q \cdot |H|$. We will show that for any $s \notin \bar{H}$, $s^{(q)}$ is not covered by the schedule $\sigma$, which will complete the proof.
Indeed, since $\sigma$ has $k \leq q $ consecutive subsequences $\sigma_1, \dots, \sigma_k$, if it would cover $s^{(q)}$, then $s$ would be fully contained within one of $\sigma_i$. But this would mean that $\sigma_i$ covers $s$, contradicting that $s \notin \bar{H}$.
\end{proof}

\begin{lemma}
    \label{lemma:VertexScheduleLemma}
    Consider a set $S$ of $k$ items. Let $\theta\in(0,1]$. Consider a sequence $\sigma$ of items from $S$, where at least $\theta k!$ permutations of the $k$ items are covered by $\sigma$. Then, there are at least $\frac{k}{2}$ items in $S$ that appear strictly more than $\theta^{\frac{2}{k}}\frac{k}{2e^2}$ times in $\sigma$.
\end{lemma}

\begin{proof}
    We start by showing a lower bound on the length of $\sigma$. Each subsequence of $k$ items can cover at most one new permutation. Therefore, we must have that
    $ \binom{|\sigma|}{k} \geq \theta k! $.
    Now use that $n!>(n/e)^n$ and $\binom{a}{b}\leq(ea/b)^b$.
    \begin{align*}
        \left(\frac{e|\sigma|}{k}\right)^{k} > \theta\left(\frac{k}{e}\right)^{k}, \quad\quad \text{ and thus } \quad
        |\sigma| &> \theta^{\frac{1}{k}}\left(\frac{k}{e}\right)^2
    \end{align*}

    Now, let $S'\subset S$ be the set of the $k/2$ items that appear the least frequently in $\sigma$. We can partition the $k!$ permutations into $(k/2)!$ subsets of permutations, where for each subset, the subsequence of the items in $S'$ has a fixed permutation. Each such partition has the same number of permutations. Suppose that fewer than $\theta (k/2)!$ permutations of the items in $S'$ are subsequences of $\sigma$. Then, no permutation of $S$ whose restriction to $S'$ is one of the remaining sub-permutations can occur, and hence fewer than a fraction $\theta$ of the permutations can occur, a contradiction. That means that for the items in $S'$, at least $\theta(k/2)!$ permutations must occur.

    Each permutation of these items also has to be a sub-sequence. Therefore, they need to occur strictly more than $\theta^{\frac{2}{k}}\left(\frac{k}{2e}\right)^2$ times, applying the above bound. Hence, there exists an item in $S'$ that occurs strictly more than $\theta^{\frac{2}{k}}\frac{k}{2e^2}$ times, by the pigeonhole principle. All the items in the other half appear at least as many times.
\end{proof}

\subsection{Lower Bounds for $0$-Stable Sparse Graphs}
\label{sec:sparse}

This section introduces a family of recursively defined graphs with $n$ vertices and whose number of edges is a parameter. Then we give a lower bound on this class of graphs, extending the prior lower bound to computation graphs with an arbitrary number of edges $m$.

\begin{definition}[Recursive Cliques]

Let $\alpha \geq 1, k \geq 3$ be two integers. A $(k,\alpha)$-recursive clique is a directed graph defined as follows. For the base level, create $k^{\alpha}$ groups, where each group is $1$ vertex. Next levels are constructed hierarchically.  A group at level $\ell$ is created from $k$ groups at level $\ell-1$. The groups are constructed recursively until there is one group.

For each group $g$ in the hierarchy, except at the top level and bottom layer, create two extra vertices, enter $e_g$ and exit $x_g$. Create an edge from $e_g$ to the enter vertex of every sub-group, and from the exit node of every sub-group to $x_g$. Additionally, for each group, create an edge from the exit vertex of every sub-group to the enter vertex of every other sub-group.
In the bottom level (a single vertex), the vertex itself is considered both the enter and exit vertex.
\end{definition}

A visualization of a $(4,2)$-recursive clique is shown in Figure~\ref{fig:graph}, and a $(3,3)$-recursive clique is shown in Figure~\ref{fig:graph33}.

\begin{figure}[H]
    \centering
    \resizebox{0.4\textwidth}{!}{\begin{tikzpicture}[
    main/.style = {circle, draw, thick, minimum size=6mm},
    enter/.style = {circle, draw, thick, fill=blue!20, minimum size=6mm},
    exit/.style = {circle, draw, thick, fill=red!20, minimum size=6mm},
    edgeclique/.style = {->, thick, color=black},
    edgeenter/.style = {->, thick, color=gray!70},
    edgeouter/.style = {->, thick, color=red}
]

\coordinate (C1) at (0,0);
\coordinate (C2) at (6,0);
\coordinate (C3) at (0,-6);
\coordinate (C4) at (6,-6);

\draw[gray, thick] (C1) circle (2.3cm);
\draw[gray, thick] (C2) circle (2.3cm);
\draw[gray, thick] (C3) circle (2.3cm);
\draw[gray, thick] (C4) circle (2.3cm);

\newcommand{\createClique}[2]{
    \foreach \j/\x/\y in {1/-0.9/0.6, 2/0.9/0.6, 3/-0.9/-0.6, 4/0.9/-0.6} {
        \node[main] (v#1\j) at ([xshift=\x cm,yshift=\y cm]#2) {};
    }
    \node[enter] (e#1) at ([yshift=1.6cm]#2) {e#1};
    \node[exit] (x#1) at ([yshift=-1.6cm]#2) {x#1};
    \foreach \j in {1,...,4} {
        \draw[edgeenter] (e#1) edge (v#1\j);
        \draw[edgeenter] (v#1\j) edge (x#1);
    }
    \foreach \a in {1,...,4} {
        \foreach \b in {1,...,4} {
            \ifnum\a<\b
                \draw[edgeclique] (v#1\a) edge (v#1\b);
                \draw[edgeclique] (v#1\b) edge (v#1\a);
            \fi
        }
    }
}

\createClique{1}{C1}
\createClique{2}{C2}
\createClique{3}{C3}
\createClique{4}{C4}

\begin{scope}[edgeouter/.append style={opacity=0.3}]
\draw[edgeouter] (x1) edge (e2);
\draw[edgeouter] (x1) edge (e3);
\draw[edgeouter] (x1) edge (e4);
\draw[edgeouter] (x2) edge (e1);
\draw[edgeouter] (x2) edge (e3);
\draw[edgeouter] (x2) edge (e4);
\draw[edgeouter, bend left=15] (x3) edge (e1);
\draw[edgeouter] (x3) edge (e2);
\draw[edgeouter] (x3) edge (e4);
\draw[edgeouter] (x4) edge (e1);
\draw[edgeouter, bend right=15] (x4) edge (e2);
\draw[edgeouter] (x4) edge (e3);
\end{scope}

\begin{scope}[even odd rule]
\clip (-3.5,-9.5) rectangle (9.5,3.5)
    (C1) circle (2.3cm) (C2) circle (2.3cm) (C3) circle (2.3cm) (C4) circle (2.3cm);
\draw[edgeouter] (x1) edge (e2);
\draw[edgeouter] (x1) edge (e3);
\draw[edgeouter] (x1) edge (e4);
\draw[edgeouter] (x2) edge (e1);
\draw[edgeouter] (x2) edge (e3);
\draw[edgeouter] (x2) edge (e4);
\draw[edgeouter, bend left=15] (x3) edge (e1);
\draw[edgeouter] (x3) edge (e2);
\draw[edgeouter] (x3) edge (e4);
\draw[edgeouter] (x4) edge (e1);
\draw[edgeouter, bend right=15] (x4) edge (e2);
\draw[edgeouter] (x4) edge (e3);
\end{scope}

\end{tikzpicture}}
    \caption{Example of a $(4,2)$-recursive clique.}
    \label{fig:graph}
\end{figure}
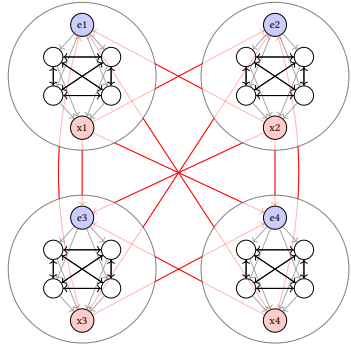

\begin{figure}[H]
    \centering
    \resizebox{0.7\textwidth}{!}{\begin{tikzpicture}[
    main/.style = {circle, draw, thick, minimum size=4mm, inner sep=0pt},
    enter/.style = {circle, draw, thick, fill=blue!20, minimum size=5mm, inner sep=0pt},
    exit/.style = {circle, draw, thick, fill=red!20, minimum size=5mm, inner sep=0pt},
    edgeclique/.style = {->, thin, color=black},
    edgeenter/.style = {->, thin, color=gray!70},
    edgemid/.style = {->, semithick, color=red},
    edgeouter/.style = {->, thick, color=green!60!black}
]

\coordinate (G1) at (8,0);
\coordinate (G2) at (0,-15);
\coordinate (G3) at (16,-15);

\def\subR{2.0}
\def\grpR{7.2}

\draw[gray, thick] (G1) circle (\grpR cm);
\draw[gray, thick] (G2) circle (\grpR cm);
\draw[gray, thick] (G3) circle (\grpR cm);

\newcommand{\placeGroup}[2]{
    \coordinate (#1C1) at ([yshift=3.6cm]#2);
    \coordinate (#1C2) at ([xshift=-3.6cm,yshift=-2.6cm]#2);
    \coordinate (#1C3) at ([xshift=3.6cm,yshift=-2.6cm]#2);
    \draw[gray!50, thin] (#1C1) circle (\subR cm);
    \draw[gray!50, thin] (#1C2) circle (\subR cm);
    \draw[gray!50, thin] (#1C3) circle (\subR cm);
}
\placeGroup{A}{G1}
\placeGroup{B}{G2}
\placeGroup{C}{G3}

\newcommand{\drawSubClique}[2]{
    \foreach \j/\x/\y in {1/0/0.8, 2/-0.85/-0.55, 3/0.85/-0.55} {
        \node[main] (#1\j) at ([xshift=\x cm,yshift=\y cm]#2) {};
    }
    \foreach \a in {1,...,3} {
        \foreach \b in {1,...,3} {
            \ifnum\a<\b
                \draw[edgeclique] (#1\a) edge (#1\b);
                \draw[edgeclique] (#1\b) edge (#1\a);
            \fi
        }
    }
}

\newcommand{\drawSubEnterExit}[3]{
    \node[enter] (#1e#2) at ([yshift=1.55cm]#3) {};
    \node[exit] (#1x#2) at ([yshift=-1.55cm]#3) {};
    \foreach \j in {1,...,3} {
        \draw[edgeenter] (#1e#2) edge (#1v#2\j);
        \draw[edgeenter] (#1v#2\j) edge (#1x#2);
    }
}

\foreach \g in {A,B,C} {
    \foreach \s in {1,2,3} {
        \drawSubClique{\g v\s}{\g C\s}
    }
}
\drawSubEnterExit{A}{1}{AC1}\drawSubEnterExit{A}{2}{AC2}\drawSubEnterExit{A}{3}{AC3}
\drawSubEnterExit{B}{1}{BC1}\drawSubEnterExit{B}{2}{BC2}\drawSubEnterExit{B}{3}{BC3}
\drawSubEnterExit{C}{1}{CC1}\drawSubEnterExit{C}{2}{CC2}\drawSubEnterExit{C}{3}{CC3}

\foreach \g/\gc in {A/G1,B/G2,C/G3} {
    \node[enter, minimum size=6mm] (\g E) at ([yshift=6.5cm]\gc) {};
    \node[exit, minimum size=6mm] (\g X) at ([yshift=-6.5cm]\gc) {};
    \foreach \s in {1,2,3} {
        \draw[edgeenter] (\g E) edge (\g e\s);
        \draw[edgeenter] (\g x\s) edge (\g X);
    }
}

\newcommand{\midEdges}[1]{
    \draw[edgemid, bend left=18] (#1x1) edge (#1e2);
    \draw[edgemid, bend left=18] (#1x1) edge (#1e3);
    \draw[edgemid, bend left=18] (#1x2) edge (#1e3);
    \draw[edgemid, bend left=18] (#1x2) edge (#1e1);
    \draw[edgemid, bend left=18] (#1x3) edge (#1e1);
    \draw[edgemid, bend left=18] (#1x3) edge (#1e2);
}
\begin{scope}[edgemid/.append style={opacity=0.25}]
\midEdges{A}\midEdges{B}\midEdges{C}
\end{scope}
\begin{scope}[even odd rule]
\clip (-7,-21) rectangle (21,7)
    (AC1) circle (\subR cm) (AC2) circle (\subR cm) (AC3) circle (\subR cm)
    (BC1) circle (\subR cm) (BC2) circle (\subR cm) (BC3) circle (\subR cm)
    (CC1) circle (\subR cm) (CC2) circle (\subR cm) (CC3) circle (\subR cm);
\midEdges{A}\midEdges{B}\midEdges{C}
\end{scope}

\newcommand{\outerEdges}{
    \draw[edgeouter, bend right=10] (AX) edge (BE);
    \draw[edgeouter, bend left=10] (AX) edge (CE);
    \draw[edgeouter, bend left=10] (BX) edge (AE);
    \draw[edgeouter, bend right=10] (BX) edge (CE);
    \draw[edgeouter, bend right=10] (CX) edge (AE);
    \draw[edgeouter, bend left=10] (CX) edge (BE);
}
\begin{scope}[edgeouter/.append style={opacity=0.3}]
\outerEdges
\end{scope}
\begin{scope}[even odd rule]
\clip (-7,-21) rectangle (21,7)
    (G1) circle (\grpR cm) (G2) circle (\grpR cm) (G3) circle (\grpR cm);
\outerEdges
\end{scope}

\end{tikzpicture}}
    \caption{Example of a $(3,3)$-recursive clique.}
    \label{fig:graph33}
\end{figure}
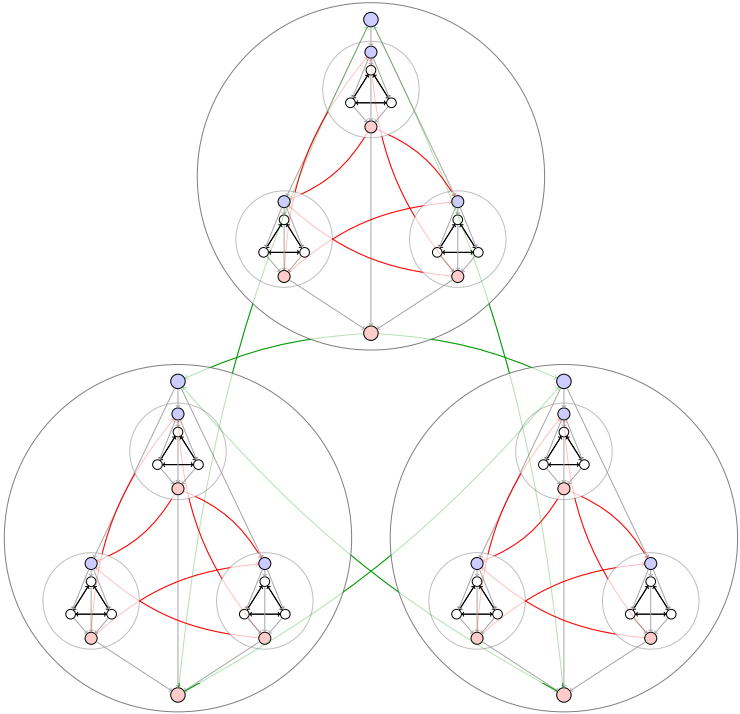

Note that the $(k,1)$-recursive clique is the standard $k$-clique.
We can control the sparsity of the graph by modifying the parameter $\alpha$; the larger $\alpha$ becomes, the sparser the graph will be. We next show some critical properties of recursive cliques.

\begin{lemma}
A $(k,\alpha)$-recursive clique has the following properties:
\begin{enumerate}
    \item the total number of groups is  $\frac{k^\alpha-1}{k-1}$;
    \item it has $\Theta(k^\alpha)$ vertices; and
    \item it has $\Theta(k^{\alpha+1})$ edges.
\end{enumerate}
\end{lemma}

\begin{proof}
 To prove the first item, note that the number of groups is just the number of nodes in a full $k$-ary tree with $\alpha$ levels, which is
 \[\sum_{i=0}^{\alpha-1}k^i=\frac{k^{\alpha}-1}{k-1}. \]

For the second item, the number of vertices is $k^\alpha$ plus two times the number of groups -- excluding the root group (each clique except the root creates an enter and exit node). Therefore, we have the following number of nodes:
    \[
    k^\alpha+2\left(\frac{k^\alpha-1}{k-1}-1\right) = \Theta( k^\alpha).
    \]

For the third item, we will count the number of edges introduced by each group. The root group introduces $k(k-1)$ edges, while all other groups add $k(k-1)+2k$ edges. This gives the following total number of edges:
    \begin{align*}
    k(k-1)+(k(k-1)+2k)\left(\frac{k^\alpha-1}{k-1}-1\right)
    = \Theta(k^{1+\alpha})
    \end{align*}
This completes the proof of the lemma.
\end{proof}

Next, we discuss the structure of simple paths in a recursive clique. Consider any group $g$ in the hierarchy above the lowest group, and let $g_1, \dots, g_k$ be its child groups. Note that the subgraphs induced by the nodes of $g_1, \dots, g_k$ are all isomorphic to each other. Let $P_g$ be the set of simple paths that start from $e_g$ (the enter node of $g$), end at $x_g$ (the exit node of $g$), and visit all nodes in the child groups. Similarly, we define $P_{g_1}, \dots, P_{g_k}$. The high-level observation is that we can create all paths in $P_g$ by choosing one path from each $P_{g_i}$ and taking all permutations of them. Formally, let
$$\mathcal{C}_g = \{\{p_1, \dots, p_k\} \mid \forall i \in [k]: p_i \in P_{g_i}\}.$$

A  set $C \in \mathcal{C}_g$ is called a {\em configuration}, and fixes a choice of one path for each child graph.
For a configuration $C$, we denote by $\textsf{Perm}(C)$ the set of all permutations of $C$.


\begin{lemma}
$P_g$ is in one-to-one correspondence with the set $\bigcup_{C \in \mathcal{C}_g} \textsf{Perm}(C)$.
\end{lemma}

\begin{proof}
Indeed, for a path $p \in P_g$ to visit all nodes in the group $g$, it has to visit all child groups $g_1, \dots, g_k$ in some order. Moreover, to visit a child group $g_i$, the path must first enter from $e_{g_i}$ and exit from $x_{g_i}$ -- but then this subpath is a path in $P_{g_i}$.

For the other direction, note that the paths for each configuration are disjoint, since at least one path is different, and can be stitched together to form a path in $P_g$ since the enter and exit nodes are connected through a clique.
\end{proof}

From now on, we will think about a permutation of a configuration in $\mathcal{C}_g$ as a path in $P_g$ and vice versa.



\begin{lemma}
Let $\theta\in(0,1]$. Suppose that for a child group $g_i$, covering $\theta_c|P_{g_i}|$ paths of $P_{g_i}$ requires length at least $L_c(\theta_c)$. Then, to cover $\theta|P_g|$ paths in $P_g$, any edge schedule needs to have length
    \[
    L_g(\theta)> \theta^{2/k}\frac{k^2}{16e^2} \cdot L_c({\theta}/{k})
    \]
\end{lemma}

\begin{proof}
If a schedule $\sigma$ covers at least $\theta|P_g|$ paths, then for at least a fraction $\theta/2$ of configurations $C$ from $\mathcal{C}_g$,  $\sigma $ covers at least a fraction $\theta/2$ of $\textsf{Perm}(C)$. Indeed, if not, then the total coverage would be strictly smaller than $\theta/2+\theta/2 = \theta$, a contradiction. Let $\mathcal{C}' \subseteq \mathcal{C}_g$ be the subset of such configurations; $|\mathcal{C'}| \geq (\theta/2) \cdot |\mathcal{C}_g|$.

For a configuration $C = \{p_1, \dots, p_k\}$, and $i \in [k]$, define by $f_{C,i}$ the number of repetitions of path $p_i$ in the schedule $\sigma$. We set a parameter $q$ to be
$$q := \left(\frac{\theta}{2}\right)^{2/k}\frac{k}{2e^2}.$$

We now apply Lemma~\ref{lemma:VertexScheduleLemma} for each configuration  $C = \{p_1, \dots, p_k\} \in \mathcal{C}'$, where each item is a path $p_i$ (note that the paths are disjoint by definition of a configuration). This implies that for at least $\frac{k}{2}$ indexes $i \in [k]$, $f_{C,i} > q$.

Define now $F_{C,i} = 1$ if $f_{C,i} \geq q$, otherwise 0. By our above result, for every $C \in \mathcal{C'}$, we have $\sum_{i \in [k]} F_{C,i} \geq \frac{k}{2}$. Thus, $\sum_{C \in \mathcal{C'}} \sum_{i \in [k]} F_{C,i} \geq \frac{k |\mathcal{C'}|}{2}$.
We claim now that there exist at least $k/4$ indexes $i \in [k]$ with $\sum_{C \in \mathcal{C'}} F_{C,i} \geq \frac{|\mathcal{C'}|}{4}$. Indeed, if not, then we would have that
$$\sum_{C \in \mathcal{C'}} \sum_{i \in [k]} F_{C,i} < \frac{k}{4} \cdot |\mathcal{C}'| + k \cdot \frac{|\mathcal{C}'|}{4} \leq \frac{k|\mathcal{C}'|}{2},$$
which is a contradiction.


Hence, for at least $k/4$ indexes $i \in [k]$, we have $\sum_{C \in \mathcal{C}_g} F_{C,i} \geq (\theta/8) \cdot |\mathcal{C}_g|$. Call the set of these indices $U$, $|U| \geq k/4$, and let $i \in U$. Note that every path $p \in P_{g_i}$ participates in exactly $|\mathcal{C}_g|/|P_{g_i}|$ configurations,  which means that there are at least ${\theta|P_{g_i}|}/8$ paths in $P_{g_i}$ for which $\sum_{C: p \in C} F_{C,i} \geq (\theta/8) \cdot |P_{g_i}|$.
%
We now apply Lemma~\ref{lemma:InductiveScheduleAppend}, where $H$ is the set of these paths and $q$ is the parameter we set. Hence, the sub-schedule that covers
the required paths has length strictly more than
    \begin{align*}
    L_c\left(\frac{\theta/8}{\left(\frac{\theta}{2}\right)^{2/k}\frac{k}{2e^2}}\right)
    \left(\frac{\theta}{2}\right)^{2/k}
    \frac{k}{2e^2}
    \geq  L_c\left(\frac{\theta}{k}\right) \cdot
    \left(\frac{\theta}{2}\right)^{2/k}
    \frac{k}{2e^2}
    \geq L_c\left(\frac{\theta}{k}\right)\frac{\theta^{2/k}k}{4e^2}
    \end{align*}
    We obtain the final result by summing the lengths of all $k/4$ sub-schedules for every $i \in U$; we can do this because the edges in each schedule are disjoint, since they involve the nodes of a different child group.
\end{proof}

\begin{lemma} \label{lem:helper:rclique}
Let $r$ be the root group in the hierarchy of a $(k,\alpha)$-recursive clique. Then, any edge schedule that covers a $\theta$-fraction of $P_r$ must have length
    \[
    |\sigma| > \psi_\alpha \cdot \theta^{2\alpha/k} \cdot k^{2\alpha+1}
    \]
 where $\psi_\alpha$ is a parameter depending only on $\alpha$.
\end{lemma}

\begin{proof}
    We will first show by induction that to cover a fraction $\theta$ of paths at level $i$ of the hierarchy, any schedule needs to have length at least
    \[
    L_i(\theta) > \frac{\theta^{2(i+1)/k}}{k^{2/k\sum_{j=0}^ij}}\frac{k^{3+2i}}{3\cdot 16^ie^{3+2i}}
    \]
    For the base case, we apply Lemma~\ref{lemma:AbsorptiveN3LB} directly, thus:
    \[
    L_0(\theta)> \theta^{2/k}\frac{k^3}{3e^3}
    \]
    For the inductive case, we have by applying the previous lemma:
    \[
    L_{i+1}(\theta)
    > \theta^{2/k}\frac{k^2}{16e^2}L_{i}\left(\frac{\theta}{k}\right)
    > \theta^{2/k}\frac{k^2}{16e^2} \frac{(\theta/k)^{2(i+1)/k}}{k^{2/k\sum_{j=0}^ij}}\frac{k^{3+2i}}{3\cdot 16^ie^{3+2i}}
    = \frac{\theta^{2(i+2)/k}}{k^{2/k(\sum_{j=0}^ij+(i+1))}}\frac{k^{3+2(i+1)}}{3\cdot 16^{i+1}e^{3+2(i+1)}}
    \]

    Finally, we apply the formula for the root layer, where $i = \alpha-1$. Hence:
    \begin{align*}
    |\sigma|
    &>
    \frac{\theta^{2\alpha/k}}{k^{2/k\sum_{j=0}^{\alpha-1}j}}\frac{k^{3+2(\alpha-1)}}{3\cdot 16^{\alpha-1}e^{3+2(\alpha-1)}}
    =\frac{\theta^{2\alpha/k}}{k^{(2/k)(\alpha-1)\alpha/2}}
    \frac{k^{2\alpha+1}}{3\cdot 16^{\alpha-1}e^{2\alpha+1}}\\
    & \geq \frac{\theta^{2\alpha/k}}{3\cdot 2^{\alpha(\alpha-1)} \cdot 16^{\alpha-1}e^{2\alpha+1}} \cdot k^{2\alpha+1}
    = \psi_\alpha \cdot \theta^{2\alpha/k} \cdot k^{2\alpha+1}
    \end{align*}
where we used the fact that $k^{1/k} \leq 2$, and $\psi_\alpha$ is a parameter depending only on $\alpha$.
\end{proof}

We now present the final piece to obtain the lower bound.

\SparseLowerBound*

\begin{proof}
Consider the class of $(k,\alpha)$-recursive cliques with $k = n^{1/\alpha}$.
We create an annotated instance $I_{k,\alpha}$ by using the color semiring. In particular, we assign a unique color to every simple path that uses all vertices of the graph, exactly as we did with the clique graph. The observation is that, similar to Lemma~\ref{lemma:clique:cover}, to have a correct schedule we need to cover all such paths. Now, we can apply Lemma~\ref{lem:helper:rclique} with $\theta=1$ to obtain the desired result.
\end{proof}

\subsection{Lower Bound for p-stable Semirings}

\label{sec:general}

Finally, we lift the lower bound to the broader class of \emph{$p$-stable} semirings.

The construction mirrors the sparse-graph instance used in the $0$-stable setting, with one crucial addition: every critical configuration must now be visited \emph{at least $p$ times}, introducing the extra multiplicative factor. To enforce this, we employ a new semiring consisting of tuples whose coordinates correspond to configurations. Each coordinate starts at $0$ and is incremented by one whenever the associated configuration is traversed, up to a cap $p$. Consequently, any relaxation algorithm that converges must raise every coordinate to its maximum value $p$, paying at least $p$ traversals per configuration and thereby inheriting the desired $p$ factor in the overall lower bound.

We note that while this is intuitive how the analysis is constructed, we must be careful that multiple configurations do not share much of the same computation to get the lower bound.



\subsubsection{Worst-Case Semirings}
We will first define the (commutative) semiring $\mathbb{S}_{p}$, $p\geq 0$, with domain $\{-\infty, 0,1,\dots, p\}$. The identity element for addition is $\mathbf{0}_{\mathbb{S}_{p}}= -\infty$ and for multiplication  $\mathbf{1}_{\mathbb{S}_{p}}= 0$. We define addition as:
\[
x \oplus_{\mathbb{S}_{p}} y :=  \max\{x,y\}
\]
and multiplication as:
\[ x \otimes_{\mathbb{S}_{p}} y :=  \min\{x+y,p\} \]
This semiring can be thought of as the tropical semiring, but the values are ``capped'' at $p$. Note that $\mathbb{S}_0$ is isomorphic to the Boolean semiring.

\begin{lemma}
$\mathbb{S}_{p}$ is a commutative semiring for any $p \geq 0$. Moreover, for any $p>0$, $\mathbb{S}_{p}$ is $p$-stable but not $(p-1)$-stable.
\end{lemma}

\begin{proof}
 It is straightforward to show that $\mathbb{S}_{p}$ is indeed a semiring. To see that it is not $(p-1)$-stable for $p>0$, observe that $1^{(p-1)} = 0+1+1^2+\dots+1^{p-1} = p-1$, but  $1^{(p)} = p$. We finally show that it is $p$-stable. Indeed, for $x=0$, we have $0^{(i)}=0$ for any $i\geq 0$. Additionally, for $x = -\infty$, we have $(-\infty)^{(i)} = 0$ for any $i\geq 0$. Finally, for $x \in \{1,\dots, p\}$, we have $x^p = p$ and $x^{(p+1)} = p$, hence it holds that $x^{(p)} = x^{(p+1)}$.
\end{proof}

We will construct the lower bound semiring by taking the direct product of multiple semirings $\mathbb{S}_{p}$. Given two semirings ${S}_1$, ${S}_2$, their {\em direct product} ${S}_1 \times {S}_2$ is a new semiring with domain the cartesian product of the domains of $S_1, S_2$, where addition and multiplication are extended componentwise to each coordinate. The direct product can be naturally extended to the product of more than two semirings. Given a semiring $S$, we define $S^i = S \times \dots \times S$ ($i$ times).

For our lower bound construction, given $N \geq 1,p\geq 0$, and letting $h$ be the number of hamiltonian cycles in the input graph, we will use the following semiring: $(\mathbb{S}_{p})^{h \cdot N}$,
in which every element in the domain is a tuple indexed by $(c,i) \in [h] \times [N]$ with its entries taking values from  $\{-\infty, 0,1,\dots, p\}$. Observe that $(\mathbb{S}_{p})^{h \cdot N}$ is also a $p$-stable semiring.

\subsubsection{The Lower Bound Instance}
We will use as the graphs for our lower bound instance $(k,\alpha)$-recursive cliques. Recall that when $\alpha=1$ the graph is just a regular clique. Let $\mathcal{P}_{k,\alpha}$ be the set of all simple paths $u \rightarrow v$ that visit all nodes in a  $(k,\alpha)$-recursive clique, and also $(v,u)$ is an edge in the graph (in other words, a Hamiltonian cycle minus one edge). Let $n$ be the number of nodes in it. Given a simple path $P\in\mathcal{P}_{k,\alpha}$, denote by $\bar{P}$ the corresponding hamiltonian cycle.

We discuss now how to construct the semiring-annotated instance $I_{k,\alpha}^p$. We create again a unique color $c(P)$ for every path $P \in \mathcal{P}_{k,\alpha}$.
For an edge $(i,j)$ in the graph, we construct the value $\mathcal{A}_{ij}$ as follows. For $P \in \mathcal{P}_{k,\alpha}$ and $\ell \in [n]$, we set the index $(c(P),\ell)$ to $1$ if $(i,j) \in \bar{P}$ and $\ell = j$, $0$ if $(i,j) \in \bar{P}$ and $\ell \neq j$, and $-\infty$ otherwise.
For a vertex $i$, we construct $\mathbf{b}_i$ as follows. For $P \in \mathcal{P}_{k,\alpha}$ and $\ell \in [n]$, we set the index $(c(P),\ell)$ to $1$ if $P$ starts at $i$ and $\ell = i$, $0$ if $P$ starts at $i$ and $\ell \neq i$, and $-\infty$ otherwise.

At a high level, the final value for a  vertex $v$ will have value $p$ at index $(c(P),\ell)$ for every $\ell \in [n]$ if the path $P$ ends at $v$, otherwise it will be $-\infty$.

\begin{lemma} \label{lem:loop:semiring}
A relaxation algorithm correctly computes the output on $I_{k,\alpha}^p$ if and only if the corresponding edge schedule $\sigma$ covers a set of walks $S$, where for all $i \in [n]$, $S$ contains for every path $P \in \mathcal{P}_{k,\alpha}$ of the form $u \rightarrow i$, the walk $(P,(i,u))^{(p)} P$.
\end{lemma}

\begin{proof}
  In other words, the set $S$ needs to repeat every Hamiltonian cycle that ends at $i$, $(p+1)$ times, with the exception of the last edge. To see that this holds, note that for the unique color $c$ of the path $P := u \rightarrow i$, the value of $\mathbf{x}_i$ must have value $p$ at the index of $(c,j)$ for every $j \in [n]$. This means that every edge in $P$ must be repeated $(p+1)$ times. But if a walk uses any other edge, the value will become $-\infty$. Hence, the only way to achieve the desired value is to loop $(p+1)$ times over the Hamiltonian cycle corresponding to the path $P$.
\end{proof}


\begin{theorem}
Consider the family of instances $I_{n,1}^p$ with $p\geq 0$ such that $p\leq \left(\frac{1}{\phi}\right)^{n/2}-1$ for some constant $\phi\in(0,1)$. Then, any edge schedule $\sigma$ that is correct for $I_{n,1}^p$ has length
$$ |\sigma| = \Omega\left(n^3(p+1)\right).$$
\end{theorem}

\begin{proof}
From Lemma~\ref{lem:loop:semiring}, any correct schedule must loop $(p+1)$ times over each path in $\mathcal{P}_{n,1}$. Hence, we can apply directly Lemma~\ref{lemma:InductiveScheduleAppend}, with $q=p+1$, $\theta_q=1$, and $H = \mathcal{P}_{n,1}$. In this case, $L(\theta)=\theta^{2/n}\frac{n^3}{3e^3}$, hence
    \[
    |\sigma| > \left(\frac{1}{p+1}\right)^{2/n}\frac{n^3}{3e^3}(p+1)
    \]
    Since $p\leq \left(\frac{1}{\phi}\right)^{n/2}-1$, we can get that
    \[
    |\sigma| > \phi\frac{n^3}{3e^3}(p+1)
    \]
    This proves the theorem.
\end{proof}

\GeneralLowerBound*

\begin{proof}
Consider the class of $(k,\alpha)$-recursive cliques with $k = n^{1/\alpha}$.
We create the annotated instance $I_{k,\alpha}^p$. From Lemma~\ref{lem:loop:semiring}, any correct schedule must loop $(p+1)$ times over each path in $\mathcal{P}_{k,\alpha}$. Hence, we can apply directly Lemma~\ref{lemma:InductiveScheduleAppend}, with $q=p+1$, $\theta_q=1$, and $H = \mathcal{P}_{k,\alpha}$. This time we have $L(\theta)= \psi_\alpha \cdot \theta^{2\alpha/k} \cdot k^{2\alpha+1}$ for some constant $\psi_\alpha$ depending on $\alpha$, by Lemma~\ref{lem:helper:rclique}. Hence,
    \[
    |\sigma| > \psi_\alpha\left(\frac{1}{p+1}\right)^{2\alpha/k}k^{2\alpha+1}(p+1)
    \]
    To get the bound we want, note that from our assumption there exists a $\phi\in(0,1)$ such that
    \[
    \left(\frac{1}{p+1}\right)^{2\alpha/k} \geq \phi
    \]
Since $n = k^{\alpha}$ and $m = k^{1+\alpha}$, this completes the proof.
\end{proof}
\fi

\end{document}